\documentclass[final]{ua-thesis}

\usepackage{verbatim}
\usepackage{amssymb,amsmath,amsthm}
\usepackage[mathscr]{eucal}

\usepackage{makeidx}
\makeindex

\newcommand*{\bbb}[1]{\mathbb{#1}}
\newcommand{\bC}{\bbb{C}}

\newcommand{\bR}{\bbb{R}}

\newcommand{\R}{\bR}
\newcommand{\Z}{\bZ}

\newcommand{\C}{\bC}
\newcommand{\F}{\mathbb{F}}

\newcommand{\rank}{\mathrm{rank}}

\newcommand{\Hom}{\mathrm{Hom}}
\DeclareMathOperator{\wt}{wt} 
\DeclareMathOperator{\im}{im} 
\DeclareMathOperator{\minwt}{minwt}

\usepackage{amsbsy}

\usepackage{ifthen}

\theoremstyle{plain}
\newtheorem{thm}[equation]{Theorem}
\newtheorem{cor}[equation]{Corollary}
\newtheorem{lem}[equation]{Lemma}
\newtheorem{prop}[equation]{Proposition}

\newtheorem*{thm*}{Theorem}
\newtheorem*{cor*}{Corollary}
\newtheorem*{lem*}{Lemma}
\newtheorem*{prop*}{Proposition}
\newtheorem*{ax*}{Axiom}

\theoremstyle{definition}
\newtheorem{defn}[equation]{Definition}

\newtheorem*{defn*}{Definition}
\newtheorem*{problem*}{Problem}
\newtheorem{ex}[equation]{Example}

\newtheorem*{ex*}{Example}
\newtheorem*{exs*}{Examples}

\newtheorem*{rem*}{Remark}

\newtheorem*{aside*}{Aside}

\newtheorem*{intuition*}{Intuition}

\newtheorem*{notn*}{Notation}

\newtheorem*{conv*}{Convention}

\newtheorem*{interp*}{Interpretation}

\newtheorem*{mnem*}{Mnemonic}

\newtheorem*{exer*}{Exercise}

\newtheorem*{conjecture*}{Conjecture}

\newcommand{\beginex}{\begin{ex}$\rhd$ }
\newcommand{\mendex}{\hfill$\lhd$\end{ex}}

\newcounter{ppart}

\def\Z{\mathbb{Z}}

\usepackage{bbm}

\usepackage[svgnames]{xcolor}
\usepackage{tikz}
\pgfdeclarelayer{backgroundlayer}
\pgfdeclarelayer{edgelayer}
\pgfdeclarelayer{nodelayer}
\pgfsetlayers{backgroundlayer,edgelayer,nodelayer,main}
\tikzstyle{none}=[inner sep=0pt]
\tikzstyle{vertex}=[circle,fill=White,draw=Black]
\tikzstyle{edge}=[rectangle,fill=White,draw=Black]
\tikzstyle{dualVertex}=[circle,fill=White,draw=Red]

\usepackage{tikz-cd}
\allowdisplaybreaks[1]

\usepackage{float}
\usepackage[bookmarks=false]{hyperref} 
\usepackage[all]{hypcap}
\usepackage{algorithm, algorithmic}

\director{Marek Rychlik}
\author{Martin Leslie}
\title{Hypermap-Homology Quantum Codes}
\date{2013}

\ifpdf
\fi

\begin{document}

\maketitle

\chapter*{Acknowledgments}
Thank you to my advisor Marek Rychlik for all his valuable effort and time and to my committee members Klaus Lux, Pham Huu Tiep and Janek Wehr. Thank you to my girlfriend Deborah Shelton, my parents Pat and Michael and my sisters Jillian and Ellen for their support over so many years. Thank you to my friends and colleagues among the graduate students in the mathematics department, especially Enrique Acosta, Michael Bishop, Michael Gilbert, Yaron Hadad, Shane Passon and Sarah Mann.

\chapter*{Dedication}
\thispagestyle{topright}
\begin{center}For my mum and dad.\end{center}

\tableofcontents
\listoffigures
\listoftables

\begin{abstract}

We introduce a new type of sparse CSS quantum error correcting code based on the homology of hypermaps. Sparse quantum error correcting codes are of interest in the building of quantum computers due to their ease of implementation and the possibility of developing fast decoders for them. Codes based on the homology of embeddings of graphs, such as Kitaev's toric code, have been discussed widely in the literature and our class of codes generalize these. We use embedded hypergraphs, which are a generalization of graphs that can have edges connected to more than two vertices.

We develop theorems and examples of our hypermap-homology codes, especially in the case that we choose a special type of basis in our homology chain complex. In particular the most straightforward generalization of the $m\times m$ toric code to hypermap-homology codes gives us a $[(3/2)m^2,2,m]$ code as compared to the toric code which is a $[2m^2,2,m]$ code. Thus we can protect the same amount of quantum information, with the same error-correcting capability, using less physical qubits.

\
\end{abstract}



\chapter{Introduction}
\label{chap:introduction}

This dissertation introduces a new type of quantum error correcting code which we will call hypermap-homology quantum codes. We also provide examples and analysis of said codes.

Quantum computers may be able to provide significant speedups over classical computers in certain problems. The most famous example is factoring integers using Shor's algorithm, where the quantum algorithm is exponentially faster than any known classical algorithm. They may also be able to simulate quantum systems faster than classical computers. However quantum noise (such as undesired unitary evolution of a quantum state or undesired measurements) makes building quantum computers difficult. Quantum error correction may be able to provide an answer to these problems. The quantum fault tolerance theorem roughly says that if we can implement quantum gates with error probability per unit of time below some threshold then by concatenating quantum codes (i.e. building sufficiently many layers of quantum error correction into our gates) we can get the error rate arbitrarily low. We will not go into the details of such practicalities, instead looking at quantum stabilizer codes as mathematical objects. See \cite{nielsenquantum} for more information on quantum computing and error correction.

We consider binary CSS (Calderbank-Shor-Steane) codes, a special case of stabilizer codes (see Chapter \ref{chap:background} for definitions and details). CSS codes are defined by two $\F_2$ matrices $H_X, H_Z$ such that $H_XH_Z^T=0$. These codes have parameters $[N,K,D]$ where $K$ information qubits are encoded into $N$ physical qubits with minimum distance $D$. A minimum distance $D$ code can correct any errors on $\lfloor(D-1)/2\rfloor$ qubits.

It was shown non-constructively in \cite{calderbank1996good} that there exist CSS codes with $K/N$ fixed and $D \sim cN$ for a constant $c$. Codes with this property are called \textit{good}. However it is not known if \em sparse \em good CSS codes (codes with $H_X$ and $H_Z$ sparse) exist. Sparse quantum codes are interesting for at least two reasons. The first is that sparse classical codes such as Low Density Parity Check (LDPC) codes are known to have performance as good as any codes while also having efficient suboptimal decoders. Being able to decode quickly is important in applications such as quantum fault tolerance. The second reason is that in physical implementations of quantum codes, the matrices can correspond to connections between the physical qubits. So a sparse quantum code can mean that the qubits are connected to few neighboring qubits which can lead to easier implementation.

 In \cite{kitaev2003fault} a family of sparse codes called toric codes were introduced which have parameters $[2m^2,2,m]$, thus $K$ is fixed and $D \sim (1/\sqrt{2})\sqrt{N}$. Despite their relatively poor performance toric and toric-like codes have promise in that they may be easy to implement and decode as discussed above. In \cite{mackay2004sparse} sparse graph quantum codes inspired by classical LDPC codes were introduced. However the most practically successful class of such codes, the `bicycle codes', are expected to have $D$ bounded above by a constant as $N$ increases.

Codes with growing distance have been suggested in the years since. We summarize some attempts with Table \ref{tab:codes}. In \cite{zemor2009cayley} it was suggested that looking at $KD^2$ may be a way to compare toric-like codes so we include this information also.

\begin{table}[h!]
\centering
\begin{tabular}{ | c | c | c | c |  }
\hline
Code & $K$ & $D$ & $KD^2$ \\
\hline
Good codes \cite{calderbank1996good} & $\sim c_1 N$ & $\sim c_2 N$ & $\sim c_3N^3$\\
Bicycle codes \cite{mackay2004sparse} & $\sim c_1N$ & expect $\leq c_2$ & $\sim c_3N$\\
Toric code \cite{kitaev2003fault} & $2$ & $(1/\sqrt{2})\sqrt{N}$ & $N$\\
Systolic freedom \cite{freedman2002z2}, \cite{fetaya2011homological} & $c_1$ & $\sim c_2\sqrt{N\log{N}}$ & $\sim c_3 N\log{N}$\\
Hypergraph-product \cite{tillich2009quantum}, \cite{kovalev2012improved} & $\sim c_1 N$ & $\sim c_2\sqrt{N}$ & $\sim c_3 N^2$\\
Cayley graph \cite{zemor2009cayley} & $\sim c_1N$ & $\sim c_2 \log{N}$ & $\sim c_3 N(\log{N})^2$ \\
Cayley graph repetition \cite{couvreur2011construction} & $\sqrt{2}\sqrt{N}$ & $(1/\sqrt{2})\sqrt{N}$ & $(1/\sqrt{2})N^{3/2}$\\
\hline
\end{tabular}
\caption[Some code constructions]{Some code constructions.}
\label{tab:codes}
\end{table}

Some results constraining the parameters of toric-like codes have been proved. In \cite{fetaya2011homological} it is shown that codes based on the homology of a fixed surface have a bound $D^2 \leq c N$. Similarly in \cite{bravyi2010tradeoffs} it is shown that for `geometrically local' codes on a 2D lattice we have $KD^2 \leq N$. Another relevant result is Gallager's proof in \cite{gallager1962low} that classical LDPC codes of column weight two cannot have minimum distance $d$ linear with blocklength $n$. This result does not imply that CSS codes with column weight two matrices cannot achieve $D\sim cN$ but it does make it quite unlikely. It is our hope (as yet unrealized) that our codes may be able to avoid these limits.

Hypergraph based quantum codes have been proposed in \cite{tillich2009quantum} (based on products of hypergraphs) and \cite{sarvepalli2012topological}. However to our knowledge no codes based on homology of hypermaps have been suggested before.

Our interest in hypermap-homology codes initially came from considering the existing construction of codes from homology of embeddings of graphs. In these constructions the fact that an edge is connected to only two vertices in classical graphs leads to one of the matrices in the CSS construction having column weight two. Hypermap-homology codes are based on embeddings of hypergraphs - a generalization of graphs where an edge can be connected to more than two vertices. Unfortunately, in our examples we make a choice of basis that does lead to the relevant matrix having column weight two. However our construction is still a generalization of codes based on graph homology and in Example \ref{ex:mbymsquaregrid} we can see that the straightforward generalization of the toric code to hypermap-homology codes has better parameters (it can store the same amount of information, with the same error correction capability, in less qubits). We see this as a proof of concept that hypermap-homology codes can be useful.

We will not discuss decoding of our codes. They are quantum LDPC codes so can be decoded by standard belief propagation techniques as in \cite{mackay2004sparse}. However the example of decoding toric codes show that this may not work particularly well. See for example \cite{duclos2010fast} for specialized techniques that may be able to be extended to hypermap-homology codes.

In Chapter 2 we give the background knowledge required in classical coding theory, quantum mechanics on qubits, and stabilizer codes. In Chapter 3 we give an exposition of how we can create CSS codes from $\F_2$-chain complexes, including from graphs embedded in surfaces and \textit{planar codes} which come from grids in the plane with certain holes removed. In Chapter 4 we give the required background in hypermaps and hypermap homology before introducing hypermap-homology codes and their analysis. Appendix A is a short discussion of computer software that we developed to compute parameters of hypermap-homology codes.

The work in this dissertation that to our knowledge is original is:
\begin{enumerate}
\item an explanation that CSS codes can be constructed from any $\F_2$-chain complex in Section \ref{sec:chainComplex},
\item an exposition of planar codes that fills in some of the more intuitive arguments from the literature in Section \ref{sec:planarCodes},
\item propositions expressing different ways to understand hypermap-homology in Section \ref{sec:hypermapHomology}, and
\item definitions and methods to determine the weight of hypermap-homology codes with a certain type of basis in Sections \ref{sec:hypHomCodes} and \ref{sec:specialBasis}. We also include examples.
\end{enumerate}

\chapter{Background}
\label{chap:background}

\section{Classical codes}
Classical information theory and coding theory were initiated by the works of Shannon and Hamming in 1948 and 1950 respectively. This section draws from a number of sources including \cite{nielsenquantum}, \cite{mackay2003information}, \cite{richardson2008modern} and \cite{guruswami2010lecture} to give an introduction to binary linear codes.

Consider the vector space $\F_2^n$ with the non-degenerate symmetric bilinear form 
\[(x,y) \mapsto x\cdot y=\sum_i x_i y_i.\]
We write elements of $\F_2^n$ as column vectors. Define an $[n,k]$ \textit{binary linear code}\index{binary linear code} $C$ to be a $k$-dimensional subspace of $\mathbb{F}_2^n$, with elements of $C$ called \textit{codewords}. We say that $G \in M_{n \times k}(\mathbb{F}_q)$ is a \textit{generator matrix}\index{generator matrix} for $C$ if $C=G(\mathbb{F}_2^k)$. A matrix $H \in M_{m \times n}(\mathbb{F}_2)$ is called a \textit{parity check matrix}\index{parity check matrix} for $C$ if $C=\ker(H)$. Since $G$ is injective we know that $G$ is full rank i.e. $\rank(G)=k$. However, since $\F_2^n/\ker(H)\cong \im(H)$ we have $\dim(\im(H))=n-k$ so with our definition $H$ does not have to be full rank. We allow the parity check matrix to include $m\geq n-k$ conditions of which $m-(n-k)$ must be redundant. This is purely for convenience; many of our code constructions will be via specifying parity check matrices which may not necessarily be full rank.

Furthermore we have $HG=0$ because columns of $G$ are elements of $C$ and $H$ times an element of $C$ is $0$.

To form a generating matrix for a code $C$ choose a basis for $C$ and place these as columns of $G$. Then we can generate all the codewords by adding basis codewords i.e. multiplying $G$ by some $x$.

To form a full rank parity check matrix for $C$ we consider the orthogonal complement 
\[C^\perp=\{y\in \mathbb{F}_2^n \colon y \cdot x = 0 \mbox{ for all } x\in C\},\]
called the \textit{dual code}\index{dual!code} of $C$. One thing to note is that we do not necessarily have $C \cap C^\perp=\{0\}$ (for example even weight codewords are orthogonal to themselves). However our bilinear form is non-degenerate so we do still have
\[\dim(C)+\dim(C^\perp)=\dim(\mathbb{F}_2^n)\]
and thus $\dim(C^\perp)=n-k$. Now choose a basis for $C^\perp$ and use this as the rows of $H$.

To show that $H$ is a parity check matrix for $C$ we need
\[C=\{x \in \mathbb{F}_2^n \colon Hx=0\}.\]
If $x \in C$ then $Hx=0$ because each row of $H$ is orthogonal to all elements of $C$. If $Hx=0$ then we know that $x$ is orthogonal to a basis for $C^\perp$ and is thus orthogonal to $C^\perp$. Thus $x\in (C^\perp)^\perp=C$ (to see this last equality note that $C \subseteq (C^\perp)^\perp$ and that $\dim(C^\perp)+\dim((C^\perp)^\perp)=n$).

We now claim that $C^\perp$ is a $[n,n-k]$ code with generator matrix $H^T$ and parity check matrix $G^T$. To see that $H^T$ is a generator matrix notice that the rows of $H$ are a basis for $C^\perp$ and thus any column vector in $C^\perp$ can be written as a sum of columns of $H^T$. For the parity check matrix: if $x \in C^\perp$ then $x$ is orthogonal to all elements of $C$ so is orthogonal to columns of $G$ so $G^T x=0$. Finally if $G^T x=0$ then $x$ is orthogonal to a basis of $C$ so is orthogonal to $C$ so $x \in C^\perp$.

We now explain error correction with an $[n,k]$ code. To transmit $u\in \F_2^k$ we send $x=Gu$. Then some error $r \in \F_2^n$ (it may be the trivial error $r=0$) occurs and we receive $y=x+r$. The decoding algorithm is to decide that the sent codeword is a codeword $x^* \in C$ which is closest to $y$. Here closest is with respect to the Hamming distance
\[d_H(a,b)=\mbox{number of elements in which } a \mbox{ and } b \mbox{ differ}.\] Then to recover the information we decode to the unique $u^*\in \F_2^k$ which corresponds to $x^*$.

With this in mind, when designing a linear code we would like codewords to be far apart. The \textit{minimum distance}\index{minimum distance!classical} $d$ of a code is the minimum of $d_H(a,b)$ for all $a \neq b \in C$. For linear codes this minimum distance is also the minimum weight of nonzero codewords where the weight of a codeword is
\[\textrm{wt}(a)=d_H(a,0)=\mbox{number of 1's in } a.\]
To see this, we have
\[d = \min_{a\neq b \in C} d_H(a,b) =\min_{a\neq b \in C} \textrm{wt}(a+b) = \min_{c \in C\setminus\{0\}} \textrm{wt}(c).\]

We will refer to such a code as an $[n,k,d]$-code.

\section{Quantum mechanics on qubits}
This introduction to quantum mechanics for quantum computing follows \cite{nielsenquantum} with some simplifications. Our quantum computing model is based on \textit{qubits}\index{qubit} (named for quantum bits) although more general qudits (quantum digits) or other systems are possible. Choose a basis of $\C^2$ (called the \textit{computational basis}) to be
\[|0\rangle = \begin{bmatrix} 1\\0 \end{bmatrix} \mbox{ and } |1\rangle=\begin{bmatrix}0\\1\end{bmatrix}.\]

We now give some postulates for the quantum mechanics\index{quantum mechanics} of qubits.
\begin{enumerate}
\item The state of $n$ qubits is an element of $\mathcal{H}_n=(\C^2)^{\otimes n}$. We work with vectors which are normalized to have $\langle \psi |\psi\rangle=1$ and consider states differing multiplicatively by $e^{i\theta}$ with $\theta\in\R$ to be equal. More concretely, $|\psi\rangle \in \mathcal{H}_n$ can be written as
\[|\psi\rangle=\sum_{i \in \F_2^n} a_i |i\rangle\]
where $|i\rangle=|i_1 i_2 \ldots i_n\rangle=|i_1\rangle |i_2\rangle\cdots |i_n\rangle=|i_1\rangle \otimes |i_2\rangle \otimes \cdots \otimes |i_n\rangle$ are all just different notations for the same thing. The normalization condition means $\sum |a_i|^2=1$.
\item The evolution of the system is described by a unitary transformation. That is, the state $|\psi\rangle$ of the system at time $t_1$ is related to the state $|\psi'\rangle$ of the system at time $t_2$ by $U=U(t_1,t_2)\in U(\mathcal{H}_n)$ via $|\psi'\rangle = U|\psi\rangle$.
\item Given an \textit{observable}\index{observable} $M$ (a Hermitian operator on $\mathcal{H}_n$) with spectral decomposition $M=\sum_m m P_m$ we can measure a state $|\psi\rangle$ with respect to $M$. The result of the measurement is $m$ with probability $p(m)=\langle\psi|P_m|\psi\rangle$ and if the result is $m$ then the new state is $P_m|\psi\rangle/\sqrt{p(m)}$.
\end{enumerate}

The following proposition justifies the commonly used statement that commuting observables can be measured `simultaneously'.

\begin{prop}
If observables $M_1,\dots,M_k$ commute then we can measure $|\psi\rangle$ with respect to these in any order and have the same probability of measurement and state outcome.
\end{prop}
\begin{proof}
Since the $M_i$ are Hermitian they are diagonalizable and since they commute they are mutually diagonalizable. Write $M_i=\sum_{m_i} m_i P_{i,m_i}$. Notice that $P_{i,m_i}P_{i,m_i'}=\delta_{m_im_i'}P_{i,m_i}$ and $P_{i,m_i}P_{j,m_j}=P_{j,m_j}P_{i,m_i}$ if $i\neq j$. 

If we first measure with respect to $M_i$ then we get result $m_i$ with probability $p(m_i)=\langle\psi|P_{i,m_i}|\psi\rangle$ and new state $|\psi'\rangle=P_{i,m_i}|\psi\rangle/\sqrt{p(m_i)}$. Measuring this new state with respect to $M_j$ gives result $m_j$ with probability 
\[p(m_j)=\langle \psi'|P_{j,m_j}|\psi'\rangle= \frac{\langle\psi|P_{i,m_i}^\dagger}{\sqrt{p(m_i)}}P_{j,m_j}\frac{P_{i,m_i}|\psi\rangle}{\sqrt{p(m_i)}}=\frac{\langle\psi|P_{j,m_j}P_{i,m_i}|\psi\rangle}{p(m_i)}\]
and new state
\[\frac{P_{j,m_j}|\psi'\rangle}{\sqrt{p(m_j)}}=\frac{P_{j,m_j}P_{i,m_i}|\psi\rangle}{\sqrt{\langle\psi|P_{j,m_j}P_{i,m_i}|\psi\rangle}}.\]
Thus measuring with respect to $M_i$ and $M_j$ in either order is equivalent to measuring with respect to $M_iM_j=\sum_{m_i,m_j}P_{j,m_j}P_{i,m_i}$. So measuring $M_1,\dots,M_k$ in any order is equivalent to measuring $M_1\dots M_k$.
\end{proof}

\section{Stabilizer codes}

The stabilizer code formalism, first introduced by Gottesman in \cite{gottesman1997stabilizer}, is a way of describing quantum codes somewhat analogous to linear codes in the classical setting. Our discussion in this section mainly follows \cite{nielsenquantum} and, in parts, \cite{mackay2004sparse}.

Define the Pauli matrices\index{Pauli matrices}
\[I=\begin{bmatrix}
1&0\\
0&1
\end{bmatrix}, \quad
X=\begin{bmatrix}
0&1\\
1&0
\end{bmatrix}, \quad Y=\begin{bmatrix}
0&-i\\
i&0
\end{bmatrix}, \quad Z=\begin{bmatrix}
1&0\\
0&-1
\end{bmatrix}.\]
They can be easily seen to be Hermitian and unitary and checked to satisfy the following relations:
\[X^2=Y^2=Z^2=I,\]
\[XY=iZ \quad ZX=iY \quad YZ=iX,\]
\[YX=-iZ \quad XZ=-iY \quad ZY=-iX.\]

Let $U(\mathcal{H}_n)$ be the group of unitary operators on the space of $n$ qubits. Define the \textit{Pauli group}\index{Pauli group} to be the group $G_n$ inside $U(\mathcal{H}_n)$ generated by operators of the form $A_1\otimes\dots\otimes A_n$ where each $A_i \in \{I,X,Y,Z\}$. We will sometimes use notation where we omit the tensor signs or include only the non-identity operators. For example $XIY=X_1Y_3 \in G_3$ is shorthand for $X \otimes I \otimes Y$. 

Using this notation we have $X_1 Y_1 Z_1=iI \in G_n$. Thus the group $G_n$ must contain $\{\pm I,\pm iI\}$. But this is enough to ensure $G_n$ is closed under products and inverses: if $c,d \in \{\pm 1,\pm i\}$ and $A_i, B_i \in \{I,X,Y,Z\}$ then we have
\[\left(c \bigotimes_{i=1}^{n} A_i\right)\left(d \bigotimes_{i=1}^{n} B_i \right) = (cd) \bigotimes_{i=1}^{n} (A_iB_i)\]
and
\[\left(c \bigotimes_{i=1}^{n} A_i\right)^{-1}=c^*\bigotimes_{i=1}^{n} A_i\]
and thus
\[G_n=\left\{c \bigotimes_{i=1}^{n} A_i \colon c\in\{\pm1,\pm i\}, A_i \in\{I,X,Y,Z\}\right\}.\]

\begin{prop} 
If we call elements of the Pauli group \em Pauli operators \em we have the following facts.
\begin{enumerate}
\item Pauli operators commute if and only if they have an even number of places with different non-identity matrices. If they do not commute then they anti-commute.
\item Squaring a Pauli operator gives $\pm I$.
\item A Pauli operator $\displaystyle c \bigotimes_{i=1}^{n} A_i$ is Hermitian if and only if $c=\pm1$.
\end{enumerate}
\end{prop}
\begin{proof}
These follow easily from our multiplication rule above and the relations among the Pauli matrices.
\end{proof}

Now for $S \leq G_n$ define $V_S \subseteq \mathcal{H}_n$ to be the set of vectors stabilized by $S$ i.e. 
\[V_S=\{|\psi\rangle \colon s |\psi\rangle = |\psi\rangle \mbox{ for all } s \in S\}.\]
It is easy to check that $V_S$ is a subspace of $\mathcal{H}_n$ and that we can write 
\[V_S=\bigcap_{s \in S} V_{\{s\}}.\]

\begin{prop}
The subspace $V_S \neq 0$ only if $-I \notin S$. In this case we have $S$ abelian, $\pm iI \notin S$ and the relations $g^2=I$ and $g^\dagger=g$ for all $g \in S$.
\end{prop}
\begin{proof}
If $-I \in S$ then $|\psi\rangle=-|\psi\rangle$ so $|\psi\rangle=0$ for all $|\psi\rangle \in V_S$.

If $S$ is not abelian then there exists $M,N \in S$ such that $MN=-NM$, then $MNM^\dagger N^\dagger=-I \in S$. If $\pm iI \in S$ then $(\pm iI)^2=-I \in S$. Also for any $g\in S$ we have $g^2=\pm I$. So $-I \notin S$ implies  $\pm iI \notin S$ and $g^2=I$ for all $g \in S$. Then $g^2=I$ implies $g^\dagger=g$ because $g$ is unitary.
\end{proof}

Define a \textit{stabilizer group}\index{stabilizer group} to be a abelian subgroup $S \leq G_n$ with $-I \notin S$. The \textit{stabilizer code}\index{stabilizer code} given by $S$ is $V_S$.

Recall the definitions of the normalizers and centralizer of a subgroup: the normalizer of the stabilizer group $S$ in $G_n$ is
\[N(S)=\{E \in G_n \colon EgE^\dagger \in S \mbox{ for all } g\in S\}\]
and the centralizer is 
\[C(S)=\{E \in G_n \colon EgE^\dagger=g\mbox{ for all } g\in S\}.\]
Clearly $C(S) \subseteq N(S)$ but in this case the inclusion is true in the opposite direction also: if $E \in N$ then $EgE^\dagger=\pm gEE^\dagger=\pm g$. So since $EgE^\dagger \in S$ we must have plus not minus and $EgE^\dagger=g$.

We will often work with generators for $S$ such that $S=\langle g_1, \ldots, g_m\rangle$. Note that $E$ commutes with all elements of $S$ if and only if it commutes with all the generators. To prove this: using the fact that the $g_l$ commute and square to one, a general $g \in S$ can be written as $g=g_1^{\epsilon_1} \cdots g_m^{\epsilon_m}$ with $\epsilon_l\in \{0,1\}$. Then if $E$ commutes with each $g_l$ it commutes with $g$.

We now describe the error correction process. Our errors are elements of $G_n$. We start with a state $|\psi\rangle \in V_S$ then after error $E$ occurs the system is in state $E|\psi\rangle$. If $S=\langle g_1, \ldots, g_{m}\rangle$ then the \textit{syndrome}\index{syndrome} of an error operator $E$ is 
\[\beta=\beta(E)=(\beta_1\ldots,\beta_{m})\] where $\beta_l \in\{0,1\}$ is defined by the equation \[Eg_l=(-1)^{\beta_l} g_lE.\]

Now the stabilizer generators $g_l$ are commuting Hermitian operators so are observables that can be measured simultaneously. Each observable $g_l$ has eigenvalues $\pm1$ (because $g_l^2=I$) and the projectors onto the $+1$ and $-1$ eigenspaces are $(I+g_l)/2$ and $(I-g_l)/2$ respectively. Thus \[g_l=(+1)\frac{I+g_l}{2}+(-1)\frac{I-g_l}{2}.\]
If $\beta_l$ is the syndrome of $g_l$ then measuring $g_l$ gives result $+$ with probability
\begin{align*}
p(+)&=\langle \psi | E^\dagger \left(\frac{I+g_l}{2}\right) E |\psi\rangle\\
&=\frac{1}{2}\langle \psi | \psi\rangle+\frac{1}{2}\langle \psi | E^\dagger g_l E |\psi\rangle\\
&=\frac{1}{2}+\frac{1}{2}(-1)^{\beta_l}\langle \psi | g_l E^\dagger E |\psi\rangle\\
&=\frac{1}{2}+\frac{1}{2}(-1)^{\beta_l}\langle \psi | \psi\rangle\\
&=\frac{1}{2}+\frac{1}{2}(-1)^{\beta_l}\\
&=1-\beta_l.
\end{align*}
Thus the outcome of the measurement is deterministic and depends only on the syndrome of the error, not the state $|\psi\rangle$.

If we have a collection of errors with distinct syndromes then we can correct the errors. In fact more is true as we will now see.

\begin{thm}Suppose $\{E_j\}$ is a set of error operators such that $E_j^\dagger E_k \notin C(S)\setminus S$ for all $j,k$. Then $\{E_j\}$ is correctable. 
\end{thm}

\begin{proof}
First note that $Eg_l=(-1)^{\beta_l}g_lE$ can be rewritten as $E^\dagger g_l=(-1)^{\beta_l}g_lE^\dagger$ so $E^\dagger$ has the same syndrome as $E$. Let $\beta(E_j)=(\beta_{j,1},\dots,\beta_{j,m})$ and similarly for $\beta(E_k)$. Then \[E_j^\dagger E_k g_l=(-1)^{\beta_{j,l}}(-1)^{\beta_{k,l}}g_l E_j^\dagger E_k\] so $E_jE_k^\dagger\in C(S)$ if and only if $(-1)^{\beta_{j,l}}(-1)^{\beta_{k,l}}=1$ for all $l$ if and only if $E_j$ and $E_k$ have the same syndrome. Thus the theorem to be proved can be restated as: a set of errors is correctable if errors with the same syndrome differ by an element of the stabilizer.

If the syndrome corresponding to $E_j$ is unique then we can correct the error by applying the operation $E_j^\dagger$, so the state becomes $E_j^\dagger E_j |\psi\rangle = |\psi\rangle$. If we have two errors $E_j$ and $E_k$ with the same syndrome then, by assumption in the theorem, $E_j^\dagger E_k \in S$. Thus even if we use the `wrong' operator, $E_j^\dagger$ instead of $E_k^\dagger$, to correct we still have $E_j^\dagger E_k |\psi\rangle = |\psi\rangle$.
\end{proof}

Define the \textit{distance}\index{minimum distance!quantum} of a stabilizer code to be the minimum weight (number of non-identity components) of a Pauli operator in $C(S) \setminus S$. Then by the theorem above a distance $d$ stabilizer code can fix errors on any $\lfloor(d-1)/2\rfloor$ qubits. We will say $V_S$ is an $[n,k,d]$ quantum code if $V_S$ is a $2^k$ dimensional subspace of $\mathcal{H}_n$ with distance $d$.

\section{Check matrix representation}
Define a function $r \colon G_n \to \F_2^{2n}$ by writing $g$ as a string of $X$'s times a string of $Z$'s, forgetting the constant out the front, and then writing those strings in binary. For example $g=XIY=i(XIX)(IIZ)$ and $r(g)=101001$. 

Notice $r(gg')=r(g)+r(g')$ and that $\ker{r}=\{\pm I,\pm iI\}$. Also for any string in $\F_2^{2n}$ we can find an operator that maps to it. Thus $G_n/\{\pm I,\pm iI\} \cong \F_2^{2n}$ (an isomorphism of a multiplicative group with the additive group of $\F_2^{2n}$). Define the \textit{effective Pauli group}\index{Pauli group!effective} to be $\overline{G}_n=G_n/\{\pm I,\pm iI\}$. We will choose coset representatives of $\overline{G}_n$ of the form
\[\bigotimes_{i=1}^{n} A_i\]
where each $A_i \in \{I,X,Y,Z\}$,

Now if $S$ is a stabilizer group then since $-I \notin S$ we have no pairs of elements differing by $-I$ or $\pm iI$ so $S\cong r(S)$.

Define \[\Lambda = \begin{bmatrix}
0&I_n\\
I_n&0
\end{bmatrix}.\]
Now write $A= \begin{bmatrix}
x(g) & z(g)
\end{bmatrix}$ and define the \textit{twisted product} 
\[r(g) \odot r(g') = r(g)\Lambda r(g')^T = x(g) \cdot z(g')+z(g) \cdot x(g').\]
Then we claim that $g$ and $g'$ commute if and only if $r(g) \odot r(g')=0$. To see this, note that wherever $g$ and $g'$ have different non-identity matrices in a position, we get exactly one one in the sum for the twisted product.

The generators for $S=\langle g_1, \ldots, g_{m}\rangle$ are said to be \textit{irredundant} if removing any generator makes the group smaller.
\begin{prop}
The generators $g_l$ of a stabilizer group are irredundant if and only if the collection of $r(g_l)$ are linearly independent.
\end{prop}
\begin{proof}
The rows are linearly independent over $\F_2$ if and only if $\sum a_i r(g_i)=0$ with $a_j =1$ for some $j$. But $\sum a_i r(g_i)=0$ if and only if $\prod g_i^{a_i} \in \{ \pm I,\pm iI\}$. But this can only be $I$ because it is in $S$ and $-I \notin S$. So some $a_j =1$ if and only if $g_j=\prod_{i\neq j} g_i^{a_i}$ i.e. $g_j$ is dependent on the other generators.
\end{proof}

So, if we create a matrix of rows which have twisted product zero and which are linearly independent then that defines a stabilizer group (choose the $g_i$ corresponding to the $r_i$ to be the one with $c=1$). If we write the matrix as
\[A= \begin{bmatrix}
A_1 & A_2
\end{bmatrix}\]
then the condition that all pairs of rows have twisted product zero becomes $A_1A_2^T+A_2A_1^T=0$.

To decode in the check matrix\index{check matrix} representation we need to find the syndrome of an error $E$. As discussed earlier this depends on whether $E$ commutes with the generators $g_l$. Thus we need to take the twisted product of rows of the check matrix $A$ with $r(E)$. In practice if we write $\begin{bmatrix}
z(E) & x(E)
\end{bmatrix}$ instead of $r(E)=\begin{bmatrix}
x(E) & z(E)
\end{bmatrix}$ then we can instead use the normal $\F_2$ matrix product. Once again, if the error has unique syndrome then it can be decoded and also if errors with the same syndrome differ by elements of the stabilizer they can be corrected.

We now develop a formula for the dimension of a stabilizer code.
\begin{lem}
If $S=\langle g_1,\dots,g_m\rangle$ with irredundant generators and $-I\notin S$ then for each $i$ there exists $h_i\in G_n$ such that $h_ig_j=(-1)^{\delta_{ij}}g_jh_i$ (i.e. $h_i$ anti commutes with $g_i$ and commutes with all other $g_j$).
\end{lem}
\begin{proof}
This is equivalent to solving $A[z(h_i)\quad x(h_i)]=e_i$ which can be done because $A$ has linearly independent rows.
\end{proof}

Now for $x\in\F_2^m$ define \[h_x=\prod_{i=1}^m h_i^{x_i} \text{ and } E_x=\{|\psi\rangle \colon g_j|\psi\rangle=(-1)^{x_j}|\psi\rangle\}.\] Notice that $h_xg_j=(-1)^{x_j}g_jh_x$.

\begin{prop}
If $S=\langle g_1,\dots,g_m\rangle$ with irredundant generators and $-I\notin S$ then $\dim V_S=2^{n-m}$.
\end{prop}
\begin{proof}
The $g_j$ are commuting and diagonalizable so we can simultaneously diagonalize them. This gives \[\mathcal{H}_n=\bigoplus_{x\in\F_2^n}E_x.\] Now we claim that $E_0 \cong E_x$ for any $x$. To see this we have the linear map $h_x\colon E_0\to E_x$. This has the correct codomain because if $|\psi\rangle\in E_0$ then $g_jh_x|\psi\rangle=(-1)^{x_j}h_x|\psi\rangle$ so $h_x|\psi\rangle\in E_x$. There is also the linear map $h_x^\dagger \colon E_x \to E_0$ for which $|\psi\rangle\in E_x$ has $g_jh_x^\dagger|\psi\rangle=(-1)^{x_j}h_x^\dagger g_j|\psi\rangle=h_x^\dagger|\psi\rangle$ so $h_x^\dagger|\psi\rangle\in E_0$. These two maps are inverses of each other and thus we have $\dim(\mathcal{H}_n)=2^m\dim E_0$. But $V_S=E_0$ so we conclude that $\dim V_S=2^{n-m}$.
\end{proof}

\section{CSS codes}

A Calderbank-Shor-Steane (CSS)\index{CSS code} code is a stabilizer code built out of two classical codes. We give a definition allowing not necessarily full rank matrices similar to the one in \cite{tillich2009quantum}.

\begin{prop}
Assume that parity check matrices $H_X$ and $H_Z$ define classical binary linear codes $C_X$ and $C_Z$ of length $n$ and that $H_XH_Z^T=0$ (this is equivalent to $C_Z^\perp \subseteq C_X$ which is equivalent to $C_X^\perp \subseteq C_Z$). Then the stabilizer code with binary check matrix
\[A=\begin{bmatrix}
H_X & 0\\
0 & H_Z
\end{bmatrix},\]
is a quantum $[n,k,d]$ code where \[k=n-\dim(C_X^\perp)-\dim(C_Z^\perp)\] and
\[d=\min \{\wt(c) \colon c \in (C_Z \setminus C_X^\perp) \cup (C_X \setminus C_Z^\perp)\}.\]
\end{prop}
\begin{proof}
We have $C_X=\ker(H_X)$ and $C_Z=\ker(H_Z)$ and thus $C_X^\perp=\im(H_X^T)$ and $C_Z^\perp=\im(H_Z^T)$. For the dimension of the code we have
\begin{align*}
k&=n-\text{\# independent rows of }A\\
&=n-\dim(\im H_Z^T)-\dim(\im H_X^T)\\
&=n-\dim(C_X^\perp)-\dim(C_Z^\perp).
\end{align*}

Next, recall that the distance of a quantum stabilizer code is
\[d=\min \{\wt(E) \colon E \in C(S)\setminus S\}.\]
Notice that
\[S=\langle X^a, Z^b \colon a \mbox{ is a row of } H_X, b \mbox{ is a row of } H_Z \rangle = \{X^aZ^b \colon a \in C_X^\perp, b\in C_Z^\perp\}\]
where $X^a$ is notation for $X_1^{a_1} \cdots X_n^{a_n}$. Now for $X^aZ^b\in C(S)$ we need $X^aZ^b$ to commute with $X^{a'}$ for $a'\in C_X^\perp$. Notice $X^aZ^bX^{a'}=(-1)^{a'\cdot b}X^{a'}X^aZ^b$ so we must have $a'\cdot b=0$ for all $a'\in C_X^\perp$ so $b\in C_X$. Similar arguments show $a\in C_Z$ so, ignoring global phase factors,
\[C(S)=\{X^aZ^b \colon a \in C_Z, b\in C_X \}.\]
For each $X^aZ^b$ at least one of $X^a$ and $X^b$ has lesser weight. Thus there exists a minimum weight vector in $C(S)\setminus S$ of the form $X^a$ or $Z^b$ so the formula in the proposition follows.
\end{proof}

\section{An example of a stabilizer code}

%
%

We describe Steane's 7 qubit code\index{Steane's 7 qubit code}. Begin with the classical $[7,4]$ Hamming code $C$ specified by the parity check matrix
\[H= \begin{bmatrix}
0 & 0 & 0 & 1 & 1 & 1 & 1\\
0 & 1 & 1 & 0 & 0 & 1 & 1 \\
1 & 0 & 1 & 0 & 1 & 0 & 1
\end{bmatrix}.\]
This code can correct one error because the syndromes of the 8 weight-zero and weight-one errors are distinct. The code also has $C^\perp\subseteq C$: the codewords in $C^\perp$ are linear combinations of rows of $H$ but the rows of $H$ are all orthogonal to all rows of $H$. Define a CSS code by $H_X=H_Z=H$. Then this is an $[n,k,d[$ code with $n=7$ and $k=7-3-3=1$. To find $d=\min\{\wt(c) \colon c\in C\setminus C^\perp\}$, note that there are no elements of weight less than or equal 2 because no column or sum of two columns of $H$ is equal to 0. Furthermore the sum of the first three columns is 0 so $1110000\in C$ and this codeword is not in $C^\perp$ because no sum of rows of $H$ is equal to it. Thus $d=3$ and we have a $[7,1,3]$ CSS code.

\chapter{Map-homology codes}
In this chapter we give an exposition of some code construction techniques from the point of view of homology. We begin with the most general case.

\section{Codes from \texorpdfstring{$\F_2$}{F2}-chain complexes}\label{sec:chainComplex}

Let $(C_\bullet, \partial_\bullet)$ be a chain complex\index{chain complex} of finite dimensional $\F_2$-vector spaces
\[\cdots \to C_{i+1} \xrightarrow{\partial_{i+1}} C_i \xrightarrow{\partial_i} C_{i-1} \to \cdots\]
(of course this means that $\partial_{i}\circ\partial_{i+1}=0$) and as usual define the homology\index{homology} vector spaces $H_i=\ker(\partial_i)/\im(\partial_{i+1})$.

We fix bases of the $C_i$ and use $[\cdot]$ to denote matrices with respect to these bases. Let $H_X=[\partial_i]$, a $\dim C_{i-1} \times \dim C_i$ matrix, and $H_Z=[\partial_{i+1}]^T$, so $H_Z$ is a $\dim C_{i+1} \times \dim C_i$ matrix. Then $H_X H_Z^T=[\partial_{i} \circ \partial_{i+1}]=0$ so we can form the $[n,k,d]$ CSS code with binary check matrix \[A=\begin{bmatrix}
H_X & 0\\
0 & H_Z
\end{bmatrix}.\] 

We have $n=\dim C_i$ and also notice $H_i=\ker(H_X)/\im(H_X^T)=C_X/C_Z^\perp$. From our discussion of classical dual codes we have $\dim(C_X)+\dim(C_X^\perp)=\dim(C_i)$. This gives
\begin{align*}
k & = \dim(C_i)-\dim(C_X^\perp)-\dim(C_Z^\perp)\\
&= \dim(C_X)-\dim(C_Z^\perp)\\
&=\dim(C_X/C_Z^\perp)\\
&=\dim(H_i).
\end{align*}

To understand the set $C_Z \setminus C_X^\perp$ we can consider the dual complex\index{dual!chain complex}. We now justify the intuitively reasonable idea that taking transposes of $H_X$ and $H_Z^T$ leads to cohomology.

Dualizing the chain complex
\[\cdots \to C_{i+1} \xrightarrow{\partial_{i+1}} C_i \xrightarrow{\partial_i} C_{i-1} \to \cdots\]
gives
\[\cdots \leftarrow C^*_{i+1} \xleftarrow{\partial^*_{i+1}} C^*_i \xleftarrow{\partial^*_i} C^*_{i-1} \leftarrow \cdots\]
where $C_i^*=\Hom(C_i,\F_2)$ and $\partial^*_i(\phi)$ is defined by $\partial^*_i(\phi)(v)=\phi(\partial_i(v))$. To check that the dual complex really is a complex we have
\[(\partial^*_{i+1}\circ\partial^*_{i})(\phi)(v)=\partial^*_{i+1}(\partial^*_i(\phi))(v)=\partial_i^*(\phi)(\partial_{i+1}(v))=\phi(\partial_i(\partial_{i+1}(v)))=0.\]
Having fixed a basis for each $C_i$ we have a non-degenerate bilinear form $\langle\,,\,\rangle\colon C_i\times C_i \to \F_2$ given by $\langle a,b\rangle=[a]\cdot[b]$. Then $C_i\cong C^*_i$ by $\psi_i\colon a\mapsto \langle a,\cdot\rangle$. Now define $\delta_i = \psi_i^{-1}\circ\partial^*_i\circ\psi_{i-1}$ and then since
\[\delta_{i+1}\circ\delta_i=\psi_{i+1}^{-1}\circ\partial^*_{i+1}\circ\psi_i\circ\psi_i^{-1}\circ\partial^*_i\circ\psi_{i-1}=\psi_{i+1}^{-1}\circ\partial^*_{i+1}\circ\partial^*_i\circ\psi_{i-1}=0\]
we have a chain complex
\[\cdots \leftarrow C_{i+1} \xleftarrow{\delta_{i+1}} C_i \xleftarrow{\delta_i} C_{i-1} \leftarrow \cdots.\]

The definition of $\delta_i$ gives us $\psi_i(\delta_i(v))=\partial_1^*(\psi_{i-1}(v))$ for all $v$ which implies $\langle\delta_i(v),w\rangle=\langle v,\partial_i (w)\rangle$ for all $v,w$. In particular if $C_i$ has basis $v_1,\dots,v_{\dim C_i}$ and $C_{i+1}$ has basis $w_1,\dots,w_{\dim C_{i+1}}$ then we have $\langle\delta_i(v_j),w_k\rangle=\langle v_j,\partial_i(w_k)\rangle$. For the right hand side of this equality we have
\[\langle v_j,\partial_i(w_k)\rangle=[v_j]\cdot[\partial_i(w_k)]=e_j\cdot [\partial_i] e_k=[\partial_i]_{j,k}\]
and similarly $\langle\delta_i(v_j),w_k\rangle=[\delta_i]_{k,j}$ which gives us $[\delta_i]=[\partial_i]^T$. With our earlier definitions $H_X=[\partial_i]$ and $H_Z=[\partial_{i+1}]^T$ we then have $[\delta_i]=H_X^T$ and $[\delta_{i+1}]=H_Z$.


Define the cohomology\index{cohomology} vector space $H^i$ by 
\[H^i=\ker(\delta_{i+1})/\im(\delta_i)=\ker(H_Z)/\im(H_X^T)=C_Z/C_X^\perp.\] Once again we have
\begin{align*}
k & = \dim(C_i)-\dim(C_X^\perp)-\dim(C_Z^\perp)\\
&= \dim(C_Z)-\dim(C_X^\perp)\\
&=\dim(C_Z/C_X^\perp)\\
&=\dim(H^i).
\end{align*}

Recall that $d=\min\{\wt(c)\colon c\in (C_X\setminus C_Z^\perp) \cup (C_Z\setminus C_X^\perp)\}$ so the minimum weight can be found by looking at elements of $C_X$ and $C_Z$ whose classes are not zero in $H_i$ and $H^i$ respectively.


\section{Codes from graphs on surfaces}\label{sec:codesFromMaps}
The idea of creating CSS codes from graphs embedded on surfaces\index{embedded graph} has been discussed in a number of papers. See for example \cite{zemor2009cayley} for an introduction to this and see \cite{lando2004graphs} for an introduction to graphs embedded in surfaces. Let $\Sigma$ be a compact, connected, oriented surface (i.e. 2-manifold) with genus $g$. We consider cellular homology with coefficients in $\F_2$ (see \cite{hatcheralgebraic} for an introduction to algebraic topology). It is well known that 
\[H_k(\Sigma)\cong\begin{cases}
\F_2 & \text{if $k=0$ or $2$}\\
(\F_2)^{2g} & \text{if $k=1$}\\
0 & \text{otherwise}.
\end{cases}\]

Define an \textit{embedding} of an undirected (simple) graph $G$ in $\Sigma$ to be a function $G\to \Sigma$ that takes vertices of $G$ to distinct points in $\Sigma$ and edges in $G$ to simple paths in $\Sigma$ (i.e. images of injective continuous functions $[0,1]\to\Sigma$) that intersect only at common vertices. We will denote such an embedding by $(\Sigma,G)$. Define a \textit{face} of the embedding to be a maximal connected subset of $\Sigma$ that does not intersect $G$. A \textit{2-cell embedding} is one where all the faces are homeomorphic to open disks; we shall only consider 2-cell embeddings. A 2-cell embedding is also called a \textit{map}\index{map|see{embedded graph}}.

Notice that a 2-cell embedding is actually a 2-dimensional CW complex and thus we have the standard cellular homology with $\F_2$ coefficients as follows. Let $V,E,F$ be sets of vertices, edges and faces of $(\Sigma,G)$ and $\mathcal{V}, \mathcal{E}, \mathcal{F}$ be $\F_2$-vector spaces with bases $V,E,F$ respectively. Then we have a chain complex
\[0 \to \mathcal{F} \xrightarrow{\partial_2} \mathcal{E} \xrightarrow{\partial_1} \mathcal{V} \to 0\]
where $\partial_2$ takes a face to the sum of the edges around that face and $\partial_1$ takes an edge to the sum of the vertices adjacent to that edge. We can then check that $\partial_1\circ \partial_2$ takes a face to twice the sum of the vertices around the face i.e. zero.

So we have $H_1((\Sigma,G))=\ker\partial_1/\im\partial_2$ and by homotopy invariance we know that $H_1((\Sigma,G)) \cong H_1(\Sigma)\cong (\F_2)^{2g}$. Create a code from the chain complex as discussed in Section \ref{sec:chainComplex}. This gives $H_X=[\partial_1]$, a $|V| \times |E|$ matrix that we can think of as the (unsigned) vertex-edge incidence matrix. Similarly $H_Z=[\partial_2]^T$ is an $|F| \times |E|$ matrix, the face-edge incidence matrix.

We now discuss the way we will understand $C_X\setminus C_Z^\perp$ in this context. Define the (Poincar\'{e}) dual\index{dual!Poincar\'{e}} of an embedded graph $(\Sigma,G)$ to be the embedded graph $(\Sigma,G^*)$ with:
\begin{itemize}
\item One vertex of $G^*$ inside each face of $G$,
\item For each edge $e$ of $G$ there is an edge $e^*$ of $G^*$ between the two vertices of $G^*$ corresponding to the two faces of $G$ adjacent to $e$.
\end{itemize}
Then the faces of $G^*$ correspond to the vertices of $G$. To see this, notice that a vertex of $G$ is adjacent to a set of edges of $G$ and the corresponding edges of $G^*$ form the boundary of a face.

So identifying the set of edges in $G^*$ with the edges of $G$ we can consider the homology chain complex for $(\Sigma,G^*)$:
 \[0 \to \mathcal{V} \xrightarrow{\delta_1} \mathcal{E} \xrightarrow{\delta_2} \mathcal{F} \to 0\]
 where $\delta_1(v)$ is the sum of edges of $G$ adjacent to $v$ and $\delta_2(e)$ is the sum of faces of $G$ adjacent to $e$. This shows that $[\delta_1]=[\partial_1]^T=H_X^T$ and $[\delta_2]=[\partial_2]^T=H_Z$ so
 \[H_1((\Sigma,G^*))\cong \ker H_Z/\im H_X^T=C_Z/C_X^\perp\cong H^1((\Sigma,G)).\]
 Also $(\Sigma,G^*)$ is a CW complex structure for $\Sigma$ so $H_1((\Sigma,G^*))\cong \F_2^{2g}$.

So from $(\Sigma,G)$ we have constructed a CSS code with parameters $[n,k,d]$ where $n$ is the number of edges of $G$, $k=2g$ and
\[d=\min\{\wt(c)\colon c\in (C_X\setminus C_Z^\perp) \cup (C_Z\setminus C_X^\perp)\}.\]
Elements of $C_X$ are cycles of the graph $G$ (equivalently subgraphs where each vertex is adjacent to an even number of edges) and similarly elements of $C_Z$ are cycles of $G^*$. Elements of $C_X\setminus C_Z^\perp$ are cycles whose homology class is not zero i.e. non-boundary cycles. So $d$ is the lowest weight (i.e. shortest) non-boundary cycle in $G$ or $G^*$.

There is a way to understand the lengths of cycles of $G^*$ just looking at $G$. A cycle in $G^*$ corresponds to a collection of faces of $G$ where each face is edge-adjacent to an even number of faces in the collection. The length of the cycle is exactly the number of edges of $G$ that we cross as we traverse the cycle. Call this the \textit{ladder distance} of the cycle.

\begin{ex}[Toric codes]
Toric codes\index{toric code} were suggested by Kitaev (see for example \cite{dennis2002topological}). Fix a positive integer $m$ and embed an $m\times m$ square grid $G$ in the torus (see Figure \ref{fig:toric} where $m=4$ and we represent the torus by a square with left and right edges identified and top and bottom edges identified). Then $G^*$ is isomorphic to $G$ (see Figure \ref{fig:toricDual} for the $m=4$ case). We have $n=2m^2$, the number of edges and $k=2g=2$. The minimum weight non-boundary cycles are the straight vertical and horizontal cycles of length $m$ so $d=m$. Thus the toric code from the $m\times m$ grid is a $[n,k,d]=[2m^2,2,m]$ code. From this we see $k=2$ is constant, $d=\sqrt{\frac{1}{2}n}$ and thus $kd^2=n$.

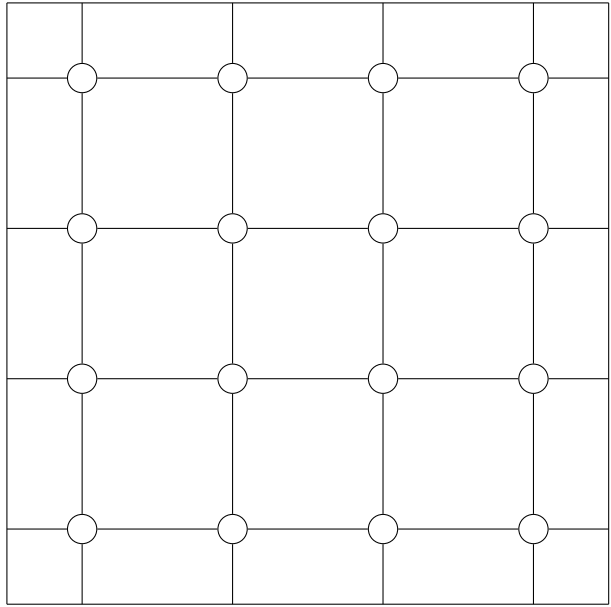
\begin{figure}[p]
\centering
\begin{tikzpicture}
	\begin{pgfonlayer}{nodelayer}
		\node [style=vertex] (0) at (-7, 6) {};
		\node [style=vertex] (1) at (-3, 6) {};
		\node [style=vertex] (2) at (-5, 4) {};
		\node [style=vertex] (3) at (-1, 4) {};
		\node [style=vertex] (4) at (-7, 2) {};
		\node [style=vertex] (5) at (-3, 2) {};
		\node [style=vertex] (6) at (-5, 0) {};
		\node [style=vertex] (7) at (-1, 0) {};
		\node [style=vertex] (8) at (-5, 6) {};
		\node [style=vertex] (9) at (-1, 6) {};
		\node [style=vertex] (10) at (-7, 4) {};
		\node [style=vertex] (11) at (-3, 4) {};
		\node [style=vertex] (12) at (-5, 2) {};
		\node [style=vertex] (13) at (-1, 2) {};
		\node [style=vertex] (14) at (-7, 0) {};
		\node [style=vertex] (15) at (-3, 0) {};
		\node [style=none] (16) at (-8, 6) {};
		\node [style=none] (17) at (-7, 7) {};
		\node [style=none] (18) at (-8, 7) {};
		\node [style=none] (19) at (-8, 4) {};
		\node [style=none] (20) at (-8, 2) {};
		\node [style=none] (21) at (-8, 0) {};
		\node [style=none] (22) at (-8, -1) {};
		\node [style=none] (23) at (-7, -1) {};
		\node [style=none] (24) at (-5, -1) {};
		\node [style=none] (25) at (-3, -1) {};
		\node [style=none] (26) at (-1, -1) {};
		\node [style=none] (27) at (0, -1) {};
		\node [style=none] (28) at (0, 0) {};
		\node [style=none] (29) at (0, 2) {};
		\node [style=none] (30) at (0, 4) {};
		\node [style=none] (31) at (0, 6) {};
		\node [style=none] (32) at (0, 7) {};
		\node [style=none] (33) at (-1, 7) {};
		\node [style=none] (34) at (-3, 7) {};
		\node [style=none] (35) at (-5, 7) {};
	\end{pgfonlayer}
	\begin{pgfonlayer}{edgelayer}
		\draw (18.center) to (22.center);
		\draw (22.center) to (27.center);
		\draw (27.center) to (32.center);
		\draw (32.center) to (18.center);
		\draw (17.center) to (0);
		\draw (0) to (16.center);
		\draw (0) to (8);
		\draw (8) to (35.center);
		\draw (34.center) to (1);
		\draw (1) to (8);
		\draw (1) to (9);
		\draw (9) to (33.center);
		\draw (9) to (31.center);
		\draw (9) to (3);
		\draw (3) to (30.center);
		\draw (3) to (11);
		\draw (11) to (1);
		\draw (11) to (2);
		\draw (2) to (8);
		\draw [in=0, out=180] (2) to (10);
		\draw (10) to (0);
		\draw (10) to (19.center);
		\draw (4) to (20.center);
		\draw (4) to (10);
		\draw (4) to (12);
		\draw (12) to (2);
		\draw (11) to (5);
		\draw (15) to (5);
		\draw (15) to (6);
		\draw (6) to (12);
		\draw (12) to (5);
		\draw (5) to (13);
		\draw (13) to (7);
		\draw (7) to (15);
		\draw (13) to (29.center);
		\draw (13) to (3);
		\draw (4) to (14);
		\draw (14) to (21.center);
		\draw (14) to (23.center);
		\draw (6) to (24.center);
		\draw (6) to (14);
		\draw (15) to (25.center);
		\draw (7) to (26.center);
		\draw (7) to (28.center);
	\end{pgfonlayer}
\end{tikzpicture}
\caption[A toric code with $m=4$]{A toric code with $m=4$.}
\label{fig:toric}
\end{figure}

\begin{figure}[p]
\centering
\begin{tikzpicture}
	\begin{pgfonlayer}{nodelayer}
		\node [style=vertex] (0) at (-7, 6) {};
		\node [style=vertex] (1) at (-3, 6) {};
		\node [style=vertex] (2) at (-5, 4) {};
		\node [style=vertex] (3) at (-1, 4) {};
		\node [style=vertex] (4) at (-7, 2) {};
		\node [style=vertex] (5) at (-3, 2) {};
		\node [style=vertex] (6) at (-5, 0) {};
		\node [style=vertex] (7) at (-1, 0) {};
		\node [style=vertex] (8) at (-5, 6) {};
		\node [style=vertex] (9) at (-1, 6) {};
		\node [style=vertex] (10) at (-7, 4) {};
		\node [style=vertex] (11) at (-3, 4) {};
		\node [style=vertex] (12) at (-5, 2) {};
		\node [style=vertex] (13) at (-1, 2) {};
		\node [style=vertex] (14) at (-7, 0) {};
		\node [style=vertex] (15) at (-3, 0) {};
		\node [style=none] (16) at (-8, 6) {};
		\node [style=none] (17) at (-7, 7) {};
		\node [style=none] (18) at (-8, 7) {};
		\node [style=none] (19) at (-8, 4) {};
		\node [style=none] (20) at (-8, 2) {};
		\node [style=none] (21) at (-8, 0) {};
		\node [style=none] (22) at (-8, -1) {};
		\node [style=none] (23) at (-7, -1) {};
		\node [style=none] (24) at (-5, -1) {};
		\node [style=none] (25) at (-3, -1) {};
		\node [style=none] (26) at (-1, -1) {};
		\node [style=none] (27) at (0, -1) {};
		\node [style=none] (28) at (0, 0) {};
		\node [style=none] (29) at (0, 2) {};
		\node [style=none] (30) at (0, 4) {};
		\node [style=none] (31) at (0, 6) {};
		\node [style=none] (32) at (0, 7) {};
		\node [style=none] (33) at (-1, 7) {};
		\node [style=none] (34) at (-3, 7) {};
		\node [style=none] (35) at (-5, 7) {};
		\node [style=dualVertex] (36) at (-6, 5) {};
		\node [style=dualVertex] (37) at (-4, 5) {};
		\node [style=dualVertex] (38) at (-2, 5) {};
		\node [style=dualVertex] (39) at (0, 5) {};
		\node [style=dualVertex] (40) at (-8, 5) {};
		\node [style=dualVertex] (41) at (-8, 3) {};
		\node [style=dualVertex] (42) at (-6, 3) {};
		\node [style=dualVertex] (43) at (-4, 3) {};
		\node [style=dualVertex] (44) at (-2, 3) {};
		\node [style=dualVertex] (45) at (0, 3) {};
		\node [style=dualVertex] (46) at (0, 1) {};
		\node [style=dualVertex] (47) at (-2, 1) {};
		\node [style=dualVertex] (48) at (-4, 1) {};
		\node [style=dualVertex] (49) at (-6, 1) {};
		\node [style=dualVertex] (50) at (-8, 1) {};
		\node [style=dualVertex] (51) at (-8, 7) {};
		\node [style=dualVertex] (52) at (-6, 7) {};
		\node [style=dualVertex] (53) at (-4, 7) {};
		\node [style=dualVertex] (54) at (-2, 7) {};
		\node [style=dualVertex] (55) at (0, 7) {};
		\node [style=dualVertex] (56) at (-8, -1) {};
		\node [style=dualVertex] (57) at (-6, -1) {};
		\node [style=dualVertex] (58) at (-4, -1) {};
		\node [style=dualVertex] (59) at (-2, -1) {};
		\node [style=dualVertex] (60) at (0, -1) {};
	\end{pgfonlayer}
	\begin{pgfonlayer}{edgelayer}
		\draw (18.center) to (22.center);
		\draw (22.center) to (27.center);
		\draw (27.center) to (32.center);
		\draw (32.center) to (18.center);
		\draw (17.center) to (0);
		\draw (0) to (16.center);
		\draw (0) to (8);
		\draw (8) to (35.center);
		\draw (34.center) to (1);
		\draw (1) to (8);
		\draw (1) to (9);
		\draw (9) to (33.center);
		\draw (9) to (31.center);
		\draw (9) to (3);
		\draw (3) to (30.center);
		\draw (3) to (11);
		\draw (11) to (1);
		\draw (11) to (2);
		\draw (2) to (8);
		\draw [in=0, out=180] (2) to (10);
		\draw (10) to (0);
		\draw (10) to (19.center);
		\draw (4) to (20.center);
		\draw (4) to (10);
		\draw (4) to (12);
		\draw (12) to (2);
		\draw (11) to (5);
		\draw (15) to (5);
		\draw (15) to (6);
		\draw (6) to (12);
		\draw (12) to (5);
		\draw (5) to (13);
		\draw (13) to (7);
		\draw (7) to (15);
		\draw (13) to (29.center);
		\draw (13) to (3);
		\draw (4) to (14);
		\draw (14) to (21.center);
		\draw (14) to (23.center);
		\draw (6) to (24.center);
		\draw (6) to (14);
		\draw (15) to (25.center);
		\draw (7) to (26.center);
		\draw (7) to (28.center);
		\draw [color=red] (36) to (37);
		\draw [color=red] (37) to (38);
		\draw [color=red] (38) to (39);
		\draw [color=red] (36) to (40);
		\draw [color=red] (36) to (42);
		\draw [color=red] (37) to (43);
		\draw [color=red] (38) to (44);
		\draw [color=red] (44) to (45);
		\draw [color=red] (44) to (43);
		\draw [color=red] (43) to (42);
		\draw [color=red] (42) to (41);
		\draw [color=red] (41) to (40);
		\draw [color=red] (40) to (18.center);
		\draw [color=red] (51) to (52);
		\draw [color=red] (52) to (53);
		\draw [color=red] (53) to (54);
		\draw [color=red] (54) to (32.center);
		\draw [color=red] (32.center) to (39);
		\draw [color=red] (39) to (45);
		\draw [color=red] (45) to (46);
		\draw [color=red] (46) to (27.center);
		\draw [color=red] (27.center) to (59);
		\draw [color=red] (59) to (58);
		\draw [color=red] (58) to (57);
		\draw [color=red] (57) to (56);
		\draw [color=red] (22.center) to (50);
		\draw [color=red] (50) to (41);
		\draw [color=red] (52) to (36);
		\draw [color=red] (53) to (37);
		\draw [color=red] (54) to (38);
		\draw [color=red] (42) to (49);
		\draw [color=red] (43) to (48);
		\draw [color=red] (44) to (47);
		\draw [color=red] (49) to (50);
		\draw [color=red] (49) to (48);
		\draw [color=red] (48) to (47);
		\draw [color=red] (47) to (46);
		\draw [color=red] (47) to (59);
		\draw [color=red] (48) to (58);
		\draw [color=red] (49) to (57);
	\end{pgfonlayer}
\end{tikzpicture}
\caption[A toric code and its dual with $m=4$]{A toric code and its dual with $m=4$. The dual is shown in red.}
\label{fig:toricDual}
\end{figure}
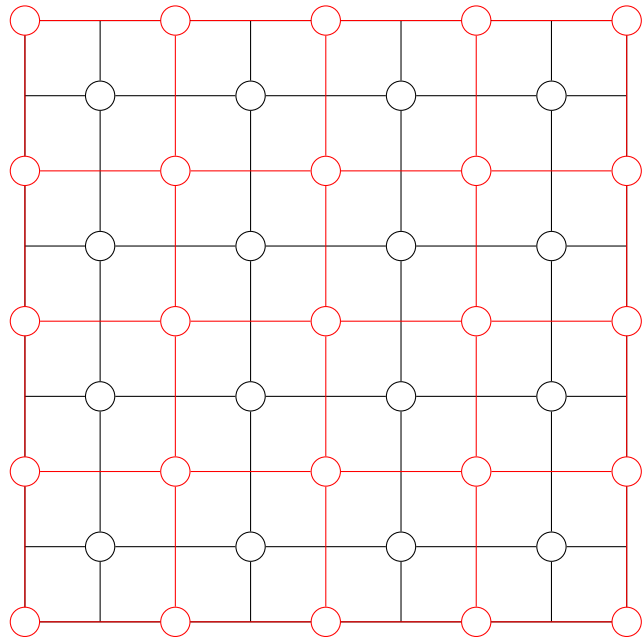

\end{ex}

\section{Planar codes}\label{sec:planarCodes}
The following is an attempt to make more explicit discussion in \cite{dennis2002topological} and \cite{bombin2007homological}. Planar codes\index{planar code} have similar properties to the toric code but allow some choice of parameters and so we can use them to show the tradeoffs in code construction from homology of graphs. They also have the practically desirable property that they can be implemented in a plane rather than a torus.

Begin with a $m \times n$ grid of vertices with $h$ `holes' (edge-contiguous collections of faces together with interior edges and vertices) removed. Think of this as a CW complex $\Gamma$ with 1-skeleton $G$ and 2-cells for each of the non-removed squares in the grid attached to $G$ by homeomorphic maps from $S^1$ to the boundary of the square. Each of these attaching maps have degree $\pm 1$. We will also define an embedding of $\tilde{G}=G$ into the sphere $S^2$ where we further attach a 2-cell to each of the removed holes and to the outer boundary of the grid.

This graph $G$ has at most $mn$ vertices and at most $(m-1)n+(n-1)m$ edges while $G$ has at most $(m-1)(n-1)$ faces and $\tilde{G}$ has at most $(m-1)(n-1)+1$ faces. The graph $\tilde{G}$ has a (Poincar\'{e}) dual graph $\tilde{G}^*$ on the sphere. As discussed in Section \ref{sec:codesFromMaps} this gives us matrices $\tilde{H}_X$ and $\tilde{H}_Z$ and since $0=H_1(S^2)=\tilde{C}_X/\tilde{C}_Z^\perp$ we have $\tilde{C}_X=\tilde{C}_Z^\perp$ and thus of course $\tilde{C}_X^\perp=\tilde{C}_Z$. In words this says all cycles of $\tilde{G}$ and $\tilde{G}^*$ are boundaries.

From the CW complex $\Gamma$ we have an $\F_2$-homology chain complex
\[0\to\mathcal{F}\to\mathcal{E}\to\mathcal{V}\to0\]
and as in Section \ref{sec:chainComplex} we construct matrices $H_X$ and $H_Z$.

Note $H_X$ is the vertex-edge incidence matrix of $G$ so $H_X=\tilde{H}_X$ and $H_Z$ is the face-edge incidence matrix of $\Gamma$ so it can be formed from $\tilde{H}_Z$, the face-edge incidence matrix of $\tilde{G}$, by removing the rows corresponding to the removed faces. These observations imply that $\tilde{C}_X=C_X$ and $\tilde{C}_Z\subseteq C_Z$.

The planar grid with $h$ holes removed is homotopy equivalent to a bouquet of $h$ circles and thus has $H_1(\Gamma)=(\F_2)^h$ and so we have an $[N,k,d]$ code with $N\leq (m-1)n+(n-1)m$ and $k=h$.

To calculate $d$ we notice that
\begin{align*}
C_X\setminus C_Z^\perp&=\tilde{C}_X\setminus C_Z^\perp\\
&=\tilde{C}_Z^\perp\setminus C_Z^\perp\\
&=\{\mbox{boundaries of }\tilde{G}\}\setminus\{\mbox{boundaries of }G\}
\end{align*}
and
\begin{align*}
C_Z\setminus C_X^\perp&=C_Z\setminus\tilde{C}_X^\perp\\
&=C_Z\setminus\tilde{C}_Z\\
&=\{\mbox{cycles of }G^*\}\setminus\{\mbox{cycles of }\tilde{G}^*\}.
\end{align*}
 
Thus a minimum weight element of $C_X\setminus C_Z^\perp$ is a boundary of a removed hole and a minimum weight element of $C_Z\setminus C_X^\perp$ is a cycle of $G^*$, which we can think of as a cycle of $\tilde{G}^*$ except we don't have to check that there an even number of edges adjacent to the vertices corresponding to removed faces. So
\begin{align*}
d=\min\{&\mbox{length of boundaries of removed faces}, \\
&\mbox{paths between holes or paths from holes to edge of grid}\}.
\end{align*}

\begin{ex}
We remove an $l\times l$ hole (including removing $(l-1)^2$ vertices and $2l(l-1)$ edges) from the middle of a square grid. See Figure \ref{fig:planarRemoveOneLbyLSquare} for an example with $l=2$. The boundary of the hole is distance $4l$ so place the hole at a ladder distance $4l$ away from the edge. Then the grid has $m=n=9l-1$ so we have a code with $N=2(9l-1)(9l-2)-2l(l-1)\sim160l^2$, $k=1$ and $d=4l$. So in this case $k=1$ is constant and $d\sim \sqrt{0.1N}$ with $kd^2 \sim 0.1N$.
\end{ex}

\begin{figure}[p]
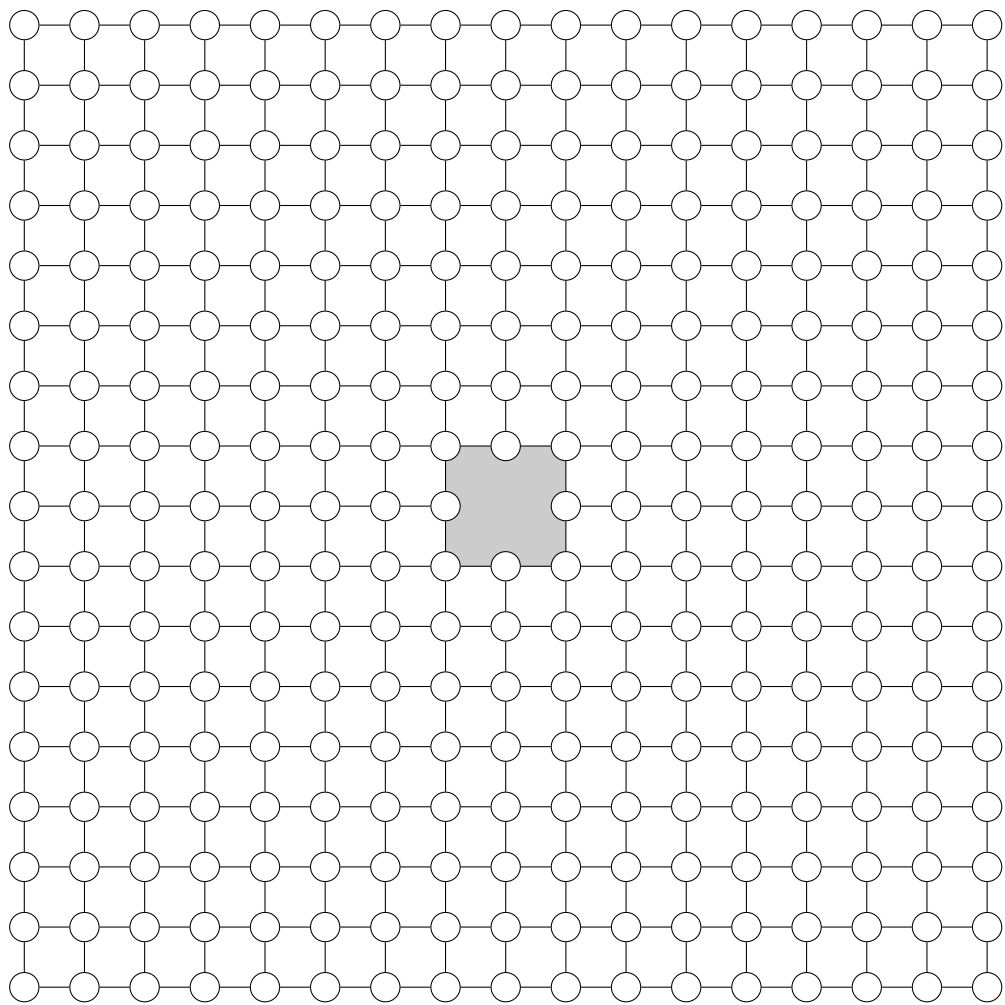

\centering

\caption[A planar code with one hole]{A planar code where we have removed one $2 \times 2$ square hole.}
\label{fig:planarRemoveOneLbyLSquare}
\end{figure}

\begin{ex}
At the opposite extreme, remove $1\times 1$ holes arranged in a $l\times l$ square grid from a larger square grid, each at distance 4 away from each other. See Figure \ref{fig:planarRemoveLbyLGridOf1by1} for an example with $l=2$. Now $m=n=4(l+1)$ so $N=2(4l+4)(4l+3)$, $k=l^2$ and $d=4$. So we get $k\sim \frac{1}{32}N$ with $d=4$ constant and $kd^2 \sim 0.5N$.
\end{ex}

\begin{figure}[htp]
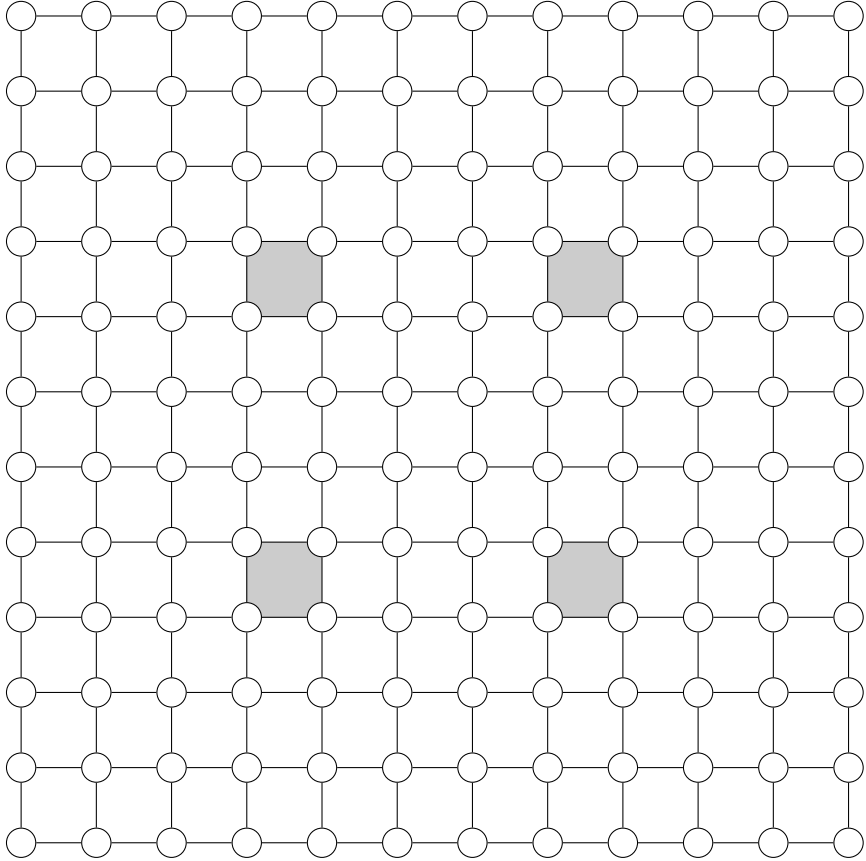

\centering
%
\caption[A planar code with a grid of holes]{A planar code where we have removed a $2 \times 2$ grid of $1\times 1$ square holes.}
\label{fig:planarRemoveLbyLGridOf1by1}
\end{figure}

\begin{ex}
Remove $l\times l$ holes arranged in $l \times l$ grid, each hole $4l$ ladder distance away from other holes and edges. See Figure \ref{fig:planarRemoveLbyLGridOfLbyL} for an example with $l=2$. Then we must have $m=n=(4l)(l+1)+l=4l^2+5l$. Then $N=2(4l^2+5l)(4l^2+5l-1)-l^2(2l(l-1))\sim 30l^4$, $k=l^2$ and $d=4l$. So in this case $k\sim\sqrt{\frac{1}{30}N}$, $d\sim\sqrt[4]{\frac{128}{15}N}$ and $kd^2 \sim \frac{8}{15}N$.
\end{ex}

\begin{figure}[p]
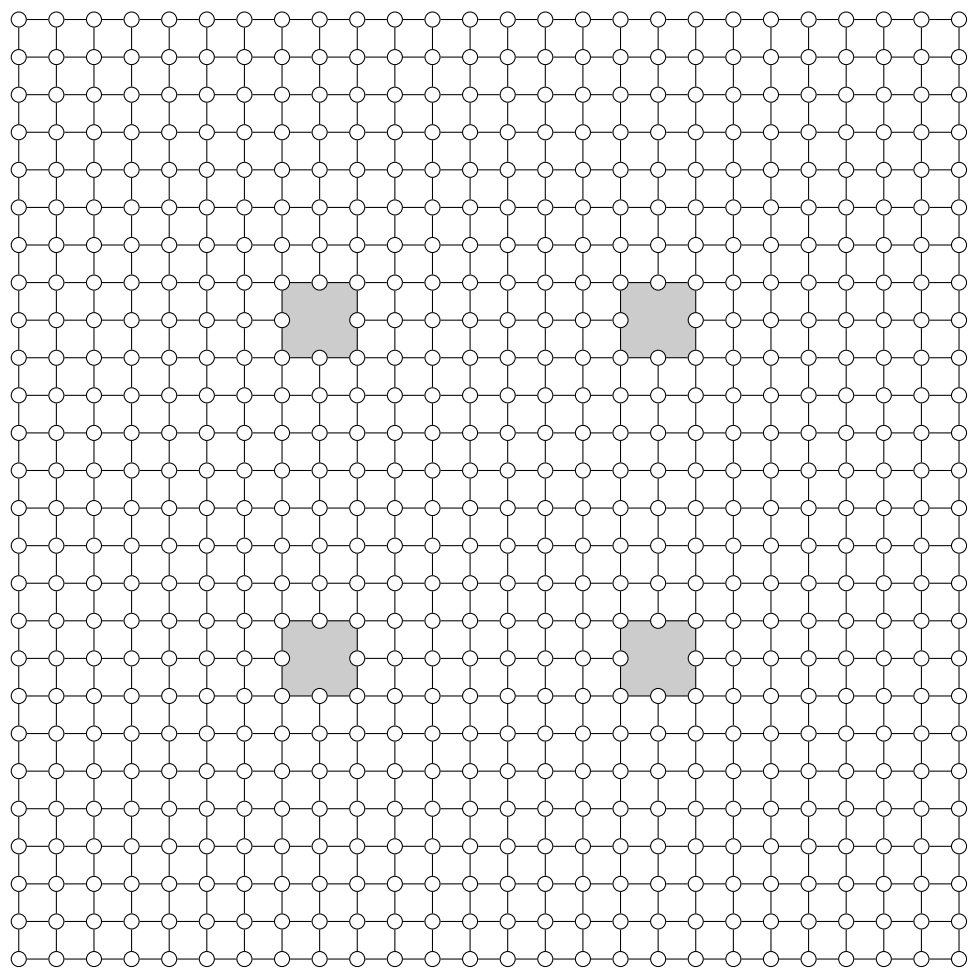

\centering

\caption[A planar code with a grid of $2\times 2$ holes]{A planar code where we have removed a $2 \times 2$ grid of $2\times 2$ square holes.}
\label{fig:planarRemoveLbyLGridOfLbyL}
\end{figure}

\chapter{Hypermap-homology codes}\label{ch:hypermapHomologyCodes}

\section{Hypermaps}

A hypergraph can be thought of as a generalization of a graph, where edges can be connected to more than two vertices. A hypermap is an embedding of a hypergraph in a surface. We will give definitions below of topological hypermaps that follow this intuition and of combinatorial hypermaps that are described instead by a pair of permutations. We will explain that they really are the same object and show how to go back and forth between them. Our discussion of hypermaps follows \cite{cori1992maps} with some differences, particularly in our choice of graphical representation. Note that, following this source, we multiply permutations left to right. We will however write the action of $\sigma$ on $i$ as $\sigma(i)$ so that with our convention we have $(\alpha\sigma)(i)=\sigma(\alpha(i))$.

\begin{defn} A \textit{hypergraph}\index{hypergraph} with $n$ darts is a pair of partitions $V$ and $E$ of $B=\{1,\ldots,n\}$. We call the elements of $V$ \textit{vertices} and the elements of $E$ \textit{edges}. Noting that each \textit{dart} (element of $B$) is in one vertex and one edge, we say that it is \textit{incident} to that vertex and that edge. A hypergraph is \textit{connected} if whenever a union of elements of $V$ equals a union of elements of $E$ then this union is empty or the whole of $B$.
\end{defn}

For a graphical representation of a hypergraph we will use the bipartite graph representation (in the context of hypermaps this is usually called the Walsh representation, after \cite{walsh1975hypermaps}). Represent vertices as circles, edges as squares and darts as line segments between the edge and vertex they are incident to. We will label darts by their number, written counterclockwise of the dart with respect to edges.

For us hypergraphs and hypermaps will always be considered to be labeled (darts, vertices, edges and faces each have a label). We will usually not label vertices, edges and faces explicitly, instead labeling them by the set of labels of adjacent darts.

Notice that a hypergraph is connected if and only if its bipartite graph representation is connected (the definition for a hypergraph to be connected is equivalent to the bipartite graph having only one connected component).

\begin{defn}
A \textit{topological (oriented) hypermap}\index{hypermap!topological} is a 2-cell embedding of the bipartite graph representation of a connected hypergraph in a compact, connected, oriented surface.
\end{defn}

Now this surface has a genus $g$ and the bipartite graph embedding satisfies Euler's formula:
\[\text	{\# vertices}-\text{\# edges}+\text{\# faces}=2-2g.\]
Let $|F|$ be the number of faces of the embedding of the bipartite graph. Then noting that the bipartite graph has $|V|+|E|$ vertices, $n$ edges and $|F|$ faces we can derive an Euler formula for topological hypermaps:
\[|V|+|E|+|F|=n+2-2g.\]

\begin{defn}
A \textit{combinatorial (oriented) hypermap}\index{hypermap!combinatorial} is a pair $(\sigma,\alpha)$ with $\sigma,\alpha \in S_n$ such that $\langle\sigma,\alpha\rangle$ is transitive on $B$.
\end{defn}

We now describe the process to go back and forth between topological and combinatorial hypermaps. Given a topological hypermap, we define $\sigma$ and $\alpha$ by $\sigma(i)$ being the dart which is counterclockwise of dart $i$ with respect to the vertex incident to $i$ and $\alpha(i)$ being the dart which is clockwise of dart $i$ with respect to the edge incident to $i$. To show $\langle\sigma,\alpha\rangle$ is transitive we note that since the hypergraph in the definition of a topological hypermap is connected there is a path from one dart to any other which we can take by using $\sigma$ and $\alpha$.

Now note that the orbits of $\sigma$ are $V$ and the orbits of $\alpha$ are $E$. Also notice $\alpha^{-1}\sigma$ goes clockwise around the interior of a face: start at a dart whose label is inside the face, $\alpha^{-1}$ takes it counterclockwise around an edge, then $\sigma$ takes ${\alpha^{-1}}(i)$ clockwise around a vertex. Thus the orbits of $\alpha^{-1}\sigma$ are the faces of the topological hypermap.

\begin{table}[h!]
\centering
\begin{tabular}{ | c | c | c |  }
\hline
Permutation & Orbits are & Ordering \\
\hline
$\sigma$ & vertices & counterclockwise\\
$\alpha$ & edges & clockwise\\
$\alpha^{-1}\sigma$ & faces & clockwise\\
\hline
\multicolumn{3}{|c|}{Labels are counterclockwise of darts w.r.t. rotation about edges.}\\
\hline
\end{tabular}
\caption[Table of conventions for hypermaps]{Table of conventions for hypermaps.}
\label{tab:conventions}
\end{table}

To go from a combinatorial hypermap to a topological hypermap we form a polygon for each cycle of $\alpha^{-1}\sigma$. If the cycle is $(i_1,\dots,i_m)$ then place the labels $i_1$,\dots,$i_m$ inside the polygon on darts that go from vertex to edge as we travel clockwise. Outside the polygon place the label ${\alpha^{-1}}(i_j)$ on the dart counterclockwise of $i_j$ with respect to rotation about edges.

Notice each label in $B$ is on the inside of precisely one polygon (because $\alpha^{-1}\sigma$ is a partition of $B$) and on the outside of precisely one polygon (because $\alpha^{-1}$ is a permutation so takes the inside partition to a partition). So we can now specify a new topological space to be the polygons glued according to their labelings (i.e. the disjoint union of the polygons modulo identifying the corresponding edges).

\begin{prop}
The space so constructed is a compact connected oriented surface and the subspace given by the now-identified edges can be thought of as the bipartite graph representation of a connected hypergraph. This construction of a topological hypermap is an inverse process to the construction of combinatorial hypermaps given above.
\end{prop}
\begin{proof}
To see that the space is a surface (i.e. a 2-manifold) we check that each point has a neighborhood homeomorphic to $\R^2$. This is certainly true inside each polygon, it is true on darts since we are gluing exactly two darts together and it is true at vertices and edges because at each one we have a cyclic order of darts around it. The surface is compact because there is a finite number of polygons. The hypergraph (i.e. the 1-skeleton) is connected because $\langle \sigma,\alpha\rangle$ is transitive and this together with the interior of the polygons being connected gives us that the surface is connected. The surface is oriented because each polygon has the clockwise orientation and so when we glue together polygons at darts, those darts have opposite orientations. The second sentence follows by construction.
\end{proof}

\begin{ex}
Consider the square below to be a torus (identify boundary: top with bottom and left with right) then interpret it as the Walsh representation of a topological hypermap.

\begin{figure}[h!]
\centering
\begin{tikzpicture}
	\begin{pgfonlayer}{nodelayer}
		\node [style=edge, minimum size=6 mm] (0) at (0, 0) {$e_1$};
		\node [style=vertex, minimum size=6 mm] (1) at (-2, 0) {$v_2$};
		\node [style=edge, minimum size=6 mm] (2) at (-2, -2) {$e_2$};
		\node [style=none] (3) at (-2, 2) {};
		\node [style=none] (4) at (0, 2) {};
		\node [style=none] (5) at (2, 0) {};
		\node [style=none] (6) at (2, -2) {};
		\node [style=none] (7) at (2, 2) {};
		\node [style=none] (8) at (2, -4) {};
		\node [style=none] (9) at (0, -4) {};
		\node [style=none] (10) at (-2, -4) {};
		\node [style=none] (11) at (-4, -4) {};
		\node [style=none] (12) at (-4, -2) {};
		\node [style=none] (13) at (-4, 0) {};
		\node [style=none] (14) at (-4, 2) {};
		\node [style=vertex, minimum size=6 mm] (15) at (0, -2) {$v_1$};
	\end{pgfonlayer}
	\begin{pgfonlayer}{edgelayer}
		\draw (1) to node[left]{5} (2);
		\draw (0) to node[below]{4} (1);
		\draw (14.center) to (11.center);
		\draw (14.center) to (7.center);
		\draw (7.center) to (8.center);
		\draw (8.center) to (11.center);
		\draw (2) to node[right]{7} (10.center);
		\draw (0) to node[above]{2} (5.center);
		\draw (0) to node[left]{1} (4.center);
		\draw (3.center) to (1);
		\draw (1) to (13.center);
		\draw [in=0, out=180] (2) to node[below]{8} (12.center);
		\draw (2) to node[above]{6} (15);
		\draw (15) to (9.center);
		\draw (15) to node[right]{3} (0);
		\draw (15) to (6.center);
	\end{pgfonlayer}
\end{tikzpicture}
\caption[First example of a hypermap]{A topological hypermap.}
\end{figure}
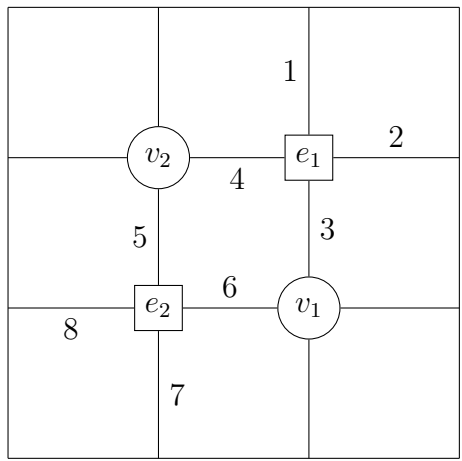

The combinatorial hypermap coming from it is $(\sigma,\alpha)$ with
\[\sigma=(1\ 8\ 3\ 6)(2\ 5\ 4\ 7), \qquad \alpha=(1\ 2\ 3\ 4)(5\ 6\ 7\ 8)\]
and from this we can calculate
\[\alpha^{-1}\sigma=(1\ 7)(2\ 8)(3\ 5)(4\ 6).\]

\end{ex}

\section{The dual hypermap}

The following definition can be found in \cite{cori1980complexity}.

\begin{defn}
The \textit{dual}\index{dual!hypermap} of a combinatorial hypermap $(\sigma,\alpha)$ is \[(\sigma',\alpha')=(\alpha^{-1}\sigma, \alpha^{-1}).\]
\end{defn}

We can see that $(\alpha')^{-1}\sigma'=(\alpha^{-1})^{-1}\alpha^{-1}\sigma=\sigma$ so duality switches vertices and faces while fixing edges (the orbits of $\alpha^{-1}$ are the same as the orbits of $\alpha$). Also $\sigma''=(\alpha')^{-1}\sigma'=\sigma$ and $\alpha''=(\alpha^{-1})^{-1}=\alpha$ so the dual of the dual is the original hypermap.

Our definition of the dual of a topological hypermap is similar to the one in \cite{mazoit2011tree}.

\begin{defn}
A \textit{dual} of a topological hypermap $H=(\Sigma,\Gamma)$ is a topological hypermap $H^*=(\Sigma_{\mathrm{op}},\Gamma^*)$ which we will now describe.
\begin{enumerate}
\item $\Sigma_{\mathrm{op}}$ is the surface $\Sigma$ with the opposite orientation.
\item The edges of $H^*$ are the edges of $H$.
\item There is precisely one vertex of $H^*$ for each face of $H$, inside that face. 
\item The darts of $H^*$ go from vertices of $H^*$ to edges around the corresponding face of $H$.
\item To label the darts of $H^*$ we draw the primal hypergraph in black and the dual hypergraph in red. Also draw a dotted line between red vertices and black vertices (this gives the \textit{canonical triangulation} of \cite{lando2004graphs}). Then each black label is inside one triangle: copy that label to the solid red line without leaving the triangle.
\end{enumerate}
\end{defn}

We will not discuss isomorphisms of combinatorial or topological hypermaps (see for example \cite{lando2004graphs}). But we do need a stronger notion of isomorphism of topological hypermaps that corresponds to equality of combinatorial hypermaps.

\begin{defn}
We say topological hypermaps $H=(\Sigma,\Gamma)$ and $H'=(\Sigma',\Gamma')$ are \textit{strongly isomorphic}, and write $H=H'$, if there exists an orientation-preserving homeomorphism $u\colon \Sigma \to \Sigma'$ with $u|_{\Gamma}$ giving an equality of hypergraphs.
\end{defn}

\begin{prop}
The dual of a topological hypermap is unique up to strong isomorphism and $(H^*)^*=H$.
\end{prop}
\begin{proof}
If $H^*=(\Sigma_{\mathrm{op}},\Gamma^*)$ and $H^*=(\Sigma_{\mathrm{op}},{\Gamma^*}')$ are both dual hypermaps of $H$ then we can see pictorially that there exists an orientation-preserving homeomorphism $u\colon \Sigma_{\mathrm{op}} \to \Sigma_{\mathrm{op}}$ with vertices and darts of $H^*$ taken to vertices and darts of ${H^*}'$. Since the labels are specified by the definition of a dual this is enough to give $H^*={H^*}'$.

Now we check that $H$ is a dual of $H^*$ and thus this uniqueness gives us $(H^*)^*=H$.
\end{proof}

\begin{prop}
The combinatorial hypermap corresponding to the topological dual of a hypermap is equal to the combinatorial dual of the hypermap.
\end{prop}
\begin{proof}
It is important to remember that the dual has the opposite orientation so all the conventions in Table \ref{tab:conventions} are reversed. Say we begin with a combinatorial hypermap $H=(\sigma,\alpha)$ which of course has combinatorial dual $H^*=(\alpha^{-1}\sigma,\alpha^{-1})$. Consider a cycle $(i_1,\dots,i_n)$ of $\alpha^{-1}\sigma$. This is a face of $H$ and so gives us a vertex of the topological dual with the darts $i_1,\dots,i_n$ clockwise around it. Our usual convention is to take vertices counterclockwise but on the dual graph we go clockwise so the vertices of $H^*$ give us exactly the permutation $\alpha^{-1}\sigma$. Similarly, cycles of $\alpha$ are edges of $H$ counterclockwise and if we take them in $H^*$ clockwise then we get the permutation $\alpha^{-1}$ as desired.
\end{proof}

\begin{ex}\label{ex:squareOctagon}

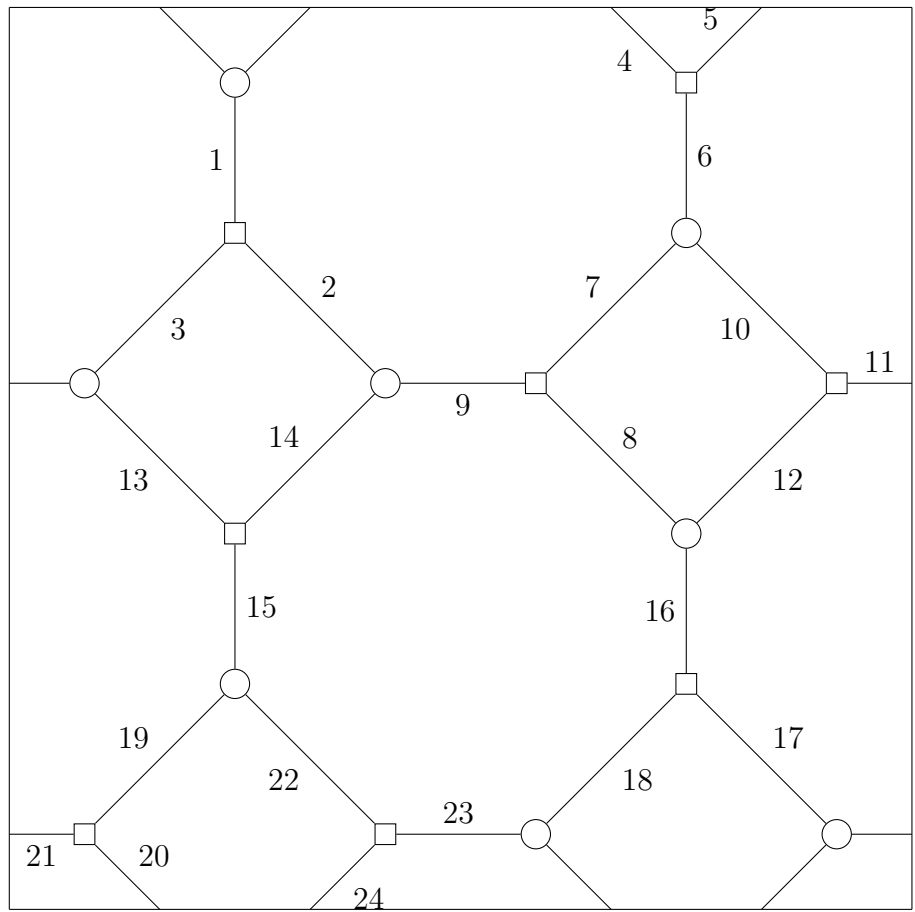
\begin{figure}[p]
\centering
\begin{tikzpicture}
	\begin{pgfonlayer}{nodelayer}
		\node [style=vertex] (0) at (-5, 5) {};
		\node [style=vertex] (1) at (-7, 1) {};
		\node [style=vertex] (2) at (-3, 1) {};
		\node [style=vertex] (3) at (-5, -3) {};
		\node [style=edge] (4) at (-5, 3) {};
		\node [style=edge] (5) at (-5, -1) {};
		\node [style=edge] (6) at (-1, 1) {};
		\node [style=edge] (7) at (1, 5) {};
		\node [style=edge] (8) at (1, -3) {};
		\node [style=edge] (9) at (-7, -5) {};
		\node [style=edge] (10) at (-3, -5) {};
		\node [style=vertex] (11) at (1, 3) {};
		\node [style=vertex] (12) at (1, -1) {};
		\node [style=vertex] (13) at (-1, -5) {};
		\node [style=vertex] (14) at (3, -5) {};
		\node [style=edge] (15) at (3, 1) {};
		\node [style=none] (16) at (-6, 6) {};
		\node [style=none] (17) at (-4, 6) {};
		\node [style=none] (18) at (0, 6) {};
		\node [style=none] (19) at (2, 6) {};
		\node [style=none] (20) at (4, 1) {};
		\node [style=none] (21) at (-8, 6) {};
		\node [style=none] (22) at (4, 6) {};
		\node [style=none] (23) at (4, -5) {};
		\node [style=none] (24) at (4, -6) {};
		\node [style=none] (25) at (2, -6) {};
		\node [style=none] (26) at (0, -6) {};
		\node [style=none] (27) at (-4, -6) {};
		\node [style=none] (28) at (-6, -6) {};
		\node [style=none] (29) at (-8, -6) {};
		\node [style=none] (30) at (-8, -5) {};
		\node [style=none] (31) at (-8, 1) {};
	\end{pgfonlayer}
	\begin{pgfonlayer}{edgelayer}
		\draw (21.center) to (22.center);
		\draw (22.center) to (24.center);
		\draw (24.center) to (29.center);
		\draw (21.center) to (29.center);
		\draw (0) to (16.center);
		\draw (0) to (17.center);
		\draw (7) to node[below left]{4} (18.center);
		\draw (7) to node[above left]{5} (19.center);
		\draw (7) to node[right]{6} (11);
		\draw (11) to node[above left]{7} (6);
		\draw (6) to node[below]{9} (2);
		\draw (2) to node[above right]{2} (4);
		\draw (4) to node[left]{1} (0);
		\draw (1) to (31.center);
		\draw (1) to node[below left]{13} (5);
		\draw (5) to node[right]{15} (3);
		\draw [in=45, out=-135] (3) to node[above left]{19} (9);
		\draw [in=0, out=180] (9) to node[below]{21} (30.center);
		\draw [in=135, out=-45] (9) to node[above right]{20} (28.center);
		\draw [in=-45, out=135] (10) to node[below left]{22} (3);
		\draw [in=45, out=-135] (10) to node[below right]{24} (27.center);
		\draw (10) to node[above]{23} (13);
		\draw (13) to (26.center);
		\draw [in=-135, out=45] (13) to node[below right]{18} (8);
		\draw [in=-135, out=45] (5) to node[above left]{14} (2);
		\draw [in=135, out=-45] (6) to node[above right]{8} (12);
		\draw [in=90, out=-90] (12) to node[left]{16} (8);
		\draw (12) to node[below right]{12} (15);
		\draw (15) to node[below left]{10} (11);
		\draw [in=180, out=0] (15) to node[above]{11} (20.center);
		\draw (8) to node[above right]{17} (14);
		\draw (14) to (25.center);
		\draw (14) to (23.center);
		\draw (1) to node[below right]{3} (4);
	\end{pgfonlayer}
\end{tikzpicture}
\caption[Octagon-square hypermap]{A hypermap with four square and four octagon faces.}
\label{fig:squareOctagon}
\end{figure}

If we start with a hypergraph embedded on a torus as pictured in Figure \ref{fig:squareOctagon} then we have a hypermap with
\begin{align*}
&\sigma=(1\ 24\ 20)(2\ 14\ 9)(3\ 11\ 13)(4\ 18\ 23)(5\ 21\ 17)(6\ 7\ 10)(8\ 16\ 12)(15\ 19\ 22),\\
&\alpha=(1\ 2\ 3)(4\ 5\ 6)(7\ 8\ 9)(10\ 11\ 12)(13\ 14\ 15)(16\ 17\ 18)(19\ 20\ 21)(22\ 23\ 24),\text{ and}\\
&\alpha^{-1}\sigma=(1\ 11\ 6\ 21)(2\ 24\ 4\ 7)(3\ 14)(5\ 18)(8\ 10)(9\ 16\ 23\ 15)(12\ 13\ 19\ 17)(20\ 22).
\end{align*}
Then, using $(\sigma',\alpha')=(\alpha^{-1}\sigma, \alpha^{-1})$, we see the combinatorial dual has 
\begin{align*}
&\sigma'=(1\ 11\ 6\ 21)(2\ 24\ 4\ 7)(3\ 14)(5\ 18)(8\ 10)(9\ 16\ 23\ 15)(12\ 13\ 19\ 17)(20\ 22)\\
&\alpha'=(1\ 3\ 2)(4\ 6\ 5)(7\ 9\ 8)(10\ 12\ 11)(13\ 15\ 14)(16\ 18\ 17)(19\ 21\ 20)(22\ 24\ 23),\text{ and}\\
&(\alpha')^{-1}\sigma'=(1\ 24\ 20)(2\ 14\ 9)(3\ 11\ 13)(4\ 18\ 23)(5\ 21\ 17)(6\ 7\ 10)(8\ 16\ 12)(15\ 19\ 22).
\end{align*}
In particular $(\alpha')^{-1}\sigma'=\sigma$. Notice that this is the same as the topological dual in Figure \ref{fig:squareOctagonDual} as long as we orient the surface in the opposite way i.e. cycle around vertices clockwise, edges counterclockwise and faces counterclockwise.

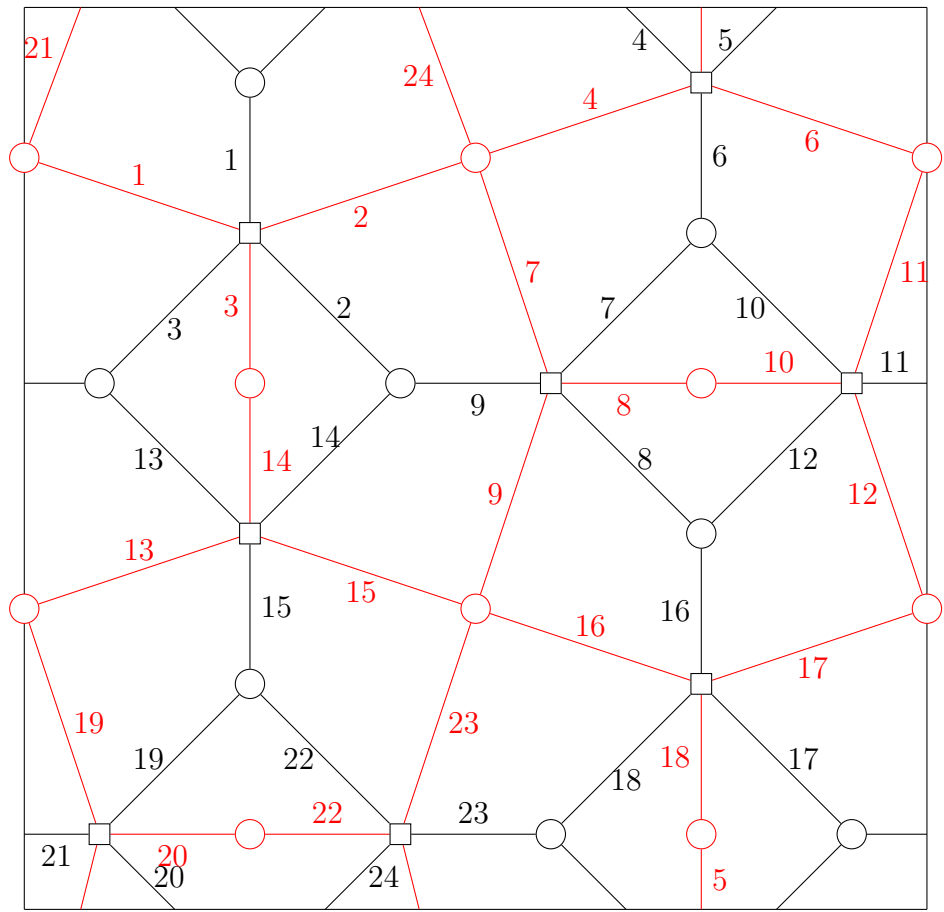
\begin{figure}[p]
\centering
\begin{tikzpicture}
	\begin{pgfonlayer}{nodelayer}
		\node [style=vertex] (0) at (-5, 5) {};
		\node [style=vertex] (1) at (-7, 1) {};
		\node [style=vertex] (2) at (-3, 1) {};
		\node [style=vertex] (3) at (-5, -3) {};
		\node [style=edge] (4) at (-5, 3) {};
		\node [style=edge] (5) at (-5, -1) {};
		\node [style=edge] (6) at (-1, 1) {};
		\node [style=edge] (7) at (1, 5) {};
		\node [style=edge] (8) at (1, -3) {};
		\node [style=edge] (9) at (-7, -5) {};
		\node [style=edge] (10) at (-3, -5) {};
		\node [style=vertex] (11) at (1, 3) {};
		\node [style=vertex] (12) at (1, -1) {};
		\node [style=vertex] (13) at (-1, -5) {};
		\node [style=vertex] (14) at (3, -5) {};
		\node [style=edge] (15) at (3, 1) {};
		\node [style=none] (16) at (-6, 6) {};
		\node [style=none] (17) at (-4, 6) {};
		\node [style=none] (18) at (0, 6) {};
		\node [style=none] (19) at (2, 6) {};
		\node [style=none] (20) at (4, 1) {};
		\node [style=none] (21) at (-8, 6) {};
		\node [style=none] (22) at (4, 6) {};
		\node [style=none] (23) at (4, -5) {};
		\node [style=none] (24) at (4, -6) {};
		\node [style=none] (25) at (2, -6) {};
		\node [style=none] (26) at (0, -6) {};
		\node [style=none] (27) at (-4, -6) {};
		\node [style=none] (28) at (-6, -6) {};
		\node [style=none] (29) at (-8, -6) {};
		\node [style=none] (30) at (-8, -5) {};
		\node [style=none] (31) at (-8, 1) {};
		\node [style=dualVertex] (32) at (1, 1) {};
		\node [style=dualVertex] (33) at (-2, 4) {};
		\node [style=dualVertex] (34) at (-5, 1) {};
		\node [style=dualVertex] (35) at (-2, -2) {};
		\node [style=dualVertex] (36) at (4, -2) {};
		\node [style=dualVertex] (37) at (-8, -2) {};
		\node [style=dualVertex] (38) at (4, 4) {};
		\node [style=dualVertex] (39) at (-8, 4) {};
		\node [style=dualVertex] (40) at (-5, -5) {};
		\node [style=dualVertex] (41) at (1, -5) {};
		\node [style=none] (42) at (1, -6) {};
		\node [style=none] (43) at (1, -6) {};
		\node [style=none] (44) at (1, 6) {};
		\node [style=none] (45) at (-5, 6) {};
		\node [style=none] (46) at (-5, -6) {};
		\node [style=none] (47) at (-2.75, 6) {};
		\node [style=none] (48) at (-2.75, 6) {};
		\node [style=none] (49) at (-2.75, -6) {};
		\node [style=none] (50) at (-7.25, 6) {};
		\node [style=none] (51) at (-7.25, 6) {};
		\node [style=none] (52) at (-7.25, -6) {};
	\end{pgfonlayer}
	\begin{pgfonlayer}{edgelayer}
		\draw (21.center) to (22.center);
		\draw (22.center) to (24.center);
		\draw (24.center) to (29.center);
		\draw (21.center) to (29.center);
		\draw (0) to (16.center);
		\draw (0) to (17.center);
		\draw (7) to node[left]{4} (18.center);
		\draw (7) to node[left]{5} (19.center);
		\draw (7) to node[right]{6} (11);
		\draw (11) to node[left]{7} (6);
		\draw (6) to node[below]{9} (2);
		\draw (2) to node[right]{2} (4);
		\draw (4) to node[left]{1} (0);
		\draw (4) to node[below]{3} (1);
		\draw (1) to (31.center);
		\draw (1) to node[left]{13} (5);
		\draw (5) to node[right]{15} (3);
		\draw [in=45, out=-135] (3) to node[left]{19} (9);
		\draw [in=0, out=180] (9) to node[below]{21} (30.center);
		\draw [in=135, out=-45] (9) to node[right]{20} (28.center);
		\draw (10) to node[left]{22} (3);
		\draw [in=45, out=-135] (10) to node[right]{24} (27.center);
		\draw (10) to node[above]{23} (13);
		\draw (13) to (26.center);
		\draw [in=-135, out=45] (13) to node[below]{18} (8);
		\draw [in=-135, out=45] (5) to node[above]{14} (2);
		\draw (6) to node[right]{8} (12);
		\draw [in=90, out=-90] (12) to node[left]{16} (8);
		\draw (12) to node[right]{12} (15);
		\draw (15) to node[left]{10} (11);
		\draw [in=180, out=0] (15) to node[above]{11} (20.center);
		\draw [in=135, out=-45] (8) to node[right]{17} (14);
		\draw (14) to (25.center);
		\draw (14) to (23.center);
		\draw [color=red] (41) to node[left]{18} (8);
		\draw [color=red, in=90, out=-90] (41) to node[right]{5} (42.center);
		\draw [color=red] (32) to node[below]{8} (6);
		\draw [color=red] (32) to node[above]{10} (15);
		\draw [color=red] (7) to (44.center);
		\draw [color=red] (33) to node[below]{2} (4);
		\draw [color=red] (33) to node[above]{4} (7);
		\draw [color=red] (33) to node[right]{7} (6);
		\draw [color=red] (35) to node[below]{15} (5);
		\draw [color=red] (35) to node[right]{23} (10);
		\draw [color=red] (35) to node[above]{16} (8);
		\draw [color=red] (35) to node[left]{9} (6);
		\draw [color=red] (34) to node[left]{3} (4);
		\draw [color=red] (34) to node[right]{14} (5);
		\draw [color=red] (39) to node[above]{1} (4);
		\draw [color=red] (38) to node[right]{11} (15);
		\draw [color=red] (38) to node[below]{6} (7);
		\draw [color=red] (5) to node[above]{13} (37);
		\draw [color=red] (37) to node[right]{19} (9);
		\draw [color=red, in=-72, out=108] (36) to node[left]{12} (15);
		\draw [color=red] (36) to node[below]{17} (8);
		\draw [color=red] (40) to node[below]{20} (9);
		\draw [color=red] (40) to node[above]{22} (10);
		\draw [color=red] (33) to node[left]{24} (47.center);
		\draw [color=red] (49.center) to (10);
		\draw [color=red] (39) to node[left, pos=0.7]{21} (50.center);
		\draw [color=red] (52.center) to (9);
	\end{pgfonlayer}
\end{tikzpicture}
\caption[Octagon-square hypermap with dual]{The octagon-square hypermap with its dual drawn in red.}
\label{fig:squareOctagonDual}
\end{figure}

\end{ex}


\section{Hypermap homology}\label{sec:hypermapHomology}
Our theory of homology of hypermaps\index{homology!hypermap} comes from \cite{cori1992maps}. Our main difference is that we are working over $\F_2$-vector spaces instead of over $\Z$-modules.

Let $\mathcal{V}, \mathcal{E}, \mathcal{F}, \mathcal{W}$ be $\F_2$-vector spaces with bases $V, E,F,W$ respectively where $W=\{w_1,\ldots,w_n\}$ . Define $d_2 \colon \mathcal{F}\to\mathcal{W}$ by $d_2(f)=\sum_{i\in f} w_i$ extended linearly and $d_1\colon \mathcal{W} \to \mathcal{V}$ by $d_1(w_i)=v_{\ni i}+v_{\ni \alpha^{-1}(i)}$. Here we use the notation $v_{\ni i}$ for the element of $V$ containing $i$ (there exists a unique such vertex because $V$ is a partition of $B$).

Also define $\imath \colon \mathcal{E} \to \mathcal{W}$ by $\imath(e)=\sum_{i\in e}w_i$ extended linearly. Notice $\imath$ is an injection ($d_2$ is also), and then define a projection map $p \colon \mathcal{W} \to \mathcal{W}/\imath(\mathcal{E})$.

\begin{prop}
We have $d_1\circ d_2=0$ and $d_1\circ \imath=0$.
\end{prop}
\begin{proof}
We have $(d_1\circ d_2)(f)=d_1(\sum_{i\in f} w_i)=\sum_{i\in f}v_{\ni i}+\sum_{i\in f}v_{\ni \alpha^{-1}(i)}$. Now both of these sums are the vertices around the face $f$ so extending linearly $d_1\circ d_2=0$. Similarly $(d_1\circ\imath)(e)=d_1(\sum_{i\in e} w_i)=\sum_{i\in e}v_{\ni i}+\sum_{i\in e}v_{\ni \alpha^{-1}(i)}$ and both of these sums are the vertices adjacent to $e$.
\end{proof}

From this we can define $\partial_1\colon \mathcal{W}/\imath(\mathcal{E}) \to \mathcal{V}$ by $\partial_1(w+\imath(\mathcal{E}))=d_1(w)$ and $\partial_2=p\circ d_2$. Then we have the commutative diagram below.
\[\begin{tikzcd}[ampersand replacement=\&]
               \& \mathcal{F} \arrow{r}{d_2} \arrow{dr}{\partial_2} \& \mathcal{W} \arrow{r}{d_1}\arrow{d}{p} \& \mathcal{V} \\
               \& \& \mathcal{W}/\imath(\mathcal{E})\arrow{ur}{\partial_1} \&
\end{tikzcd}\]
From this we see $\partial_1 \circ \partial_2 = d_1\circ d_2=0$. So we have a chain complex
\[\mathcal{F} \stackrel{\partial_2}{\to} \mathcal{W}/\imath(\mathcal{E}) \stackrel{\partial_1}{\to} \mathcal{V}\]
\[V_S=\cap_{s \in S} V_{\{s\}}.\]

First we see that if $\mathcal{V}$ has basis $v_1,\dots,v_{|V|}$ then $\im\partial_1$ has basis $v_1-v_j$ for $j=2,\dots,|V|$. This set is linearly independent because the $v_j$ are a basis of $\mathcal{V}$ and it spans $\im\partial_1$ because $\partial_1(w_i)=v_{\ni i}+v_{\ni \alpha^{-1}(i)}=(v_1+v_{\ni \alpha^{-1}(i)})+(v_1+v_{\ni i})$. Thus $\dim(\im\partial_1)=|V|-1$ and $\dim H_0=1$.

Next $\ker\partial_2=\{f \in \mathcal{F} \colon d_2(f)\in \imath(\mathcal{E})\}$ which has the same dimension as $d_2(\mathcal{F}) \cap \imath(\mathcal{E})$ because $d_2$ is injective. But an element in this intersection corresponds to both a union of cycles of $\alpha^{-1}\sigma$ and a union of cycles of $\alpha$. But such a union must be empty or the whole of $B$ and so $\dim H_2=\dim(\ker\partial_2)=1$.

From this we have $\dim(\im\partial_2)=\dim\mathcal{F}-\dim(\ker\partial_2)=|F|-1$ and $\dim(\ker\partial_1)=\dim(\mathcal{W}/\imath(\mathcal{E}))-\dim(\im\partial_1)=|W|-|E|-(|V|-1)=n-|V|-|E|+1$. So then $\dim H_1=\dim(\ker\partial_1)-\dim(\im\partial_2)=(n-|V|-|E|+1)-(|F|-1)=n+2-|V|-|E|-|F|=2g$.

The following well known lemma allows us to create a homomorphism of homology groups induced from a homomorphism of chains.

\begin{lem}\label{lem:inducedHomomorphism}
If we have a commutative diagram of vector spaces and vector space homomorphisms (i.e. linear functions)
\[\begin{tikzcd}[ampersand replacement=\&]
\& A_2 \arrow{r}{f_2} \arrow{d}{\psi_2} \& A_1 \arrow{r}{f_1}\arrow{d}{\psi_1} \& A_0 \arrow{d}{\psi_0}\\
\&B_2 \arrow{r}{g_2} \& B_1 \arrow{r}{g_1} \& B_0 
\end{tikzcd}\]
with $f_1\circ f_2=0$ and $g_1\circ g_2=0$ then $\psi_1(\im f_2) \subseteq \im g_2$, $\psi_1(\ker f_1)\subseteq \ker g_1$ which gives us a homomorphism\index{induced homomorphism}
\[(\psi_1)_* \colon \frac{\ker f_1}{\im f_2} \to \frac{\ker g_1}{\im g_2}.\]
\end{lem}
\begin{proof}
An element of $\psi_1(\im f_2)$ is of the form $\psi_1(f_2(a_2))=g_2(\psi_2(a_2))\in \im g_2$. Next an element of $\psi_1(\ker f_1)$ is of the form $\psi_1(a_1)$ with $f_1(a_1)=0$. But then $g_1(\psi_1(a_1))=\psi_0(f_1(a_1))=\psi_0(0)=0$ so $\psi_1(a_1) \in \ker g_1$.

Then we can define $(\psi_1)_*(a_1+\im f_2)=\psi_1(a_1)+\im g_2$ and the facts above give us that this function is well defined and has an allowable codomain. Linearity follows from the linearity of $\psi_1$.
\end{proof}

The next lemma gives us a condition on the above diagram that lead to an isomorphism of homology.

\begin{lem}\label{lem:inducedIsomorphism}
If we have the setup in Lemma \ref{lem:inducedHomomorphism} and furthermore we have \[\psi_1(\im f_2) = \im g_2,\] \[\psi_1(\ker f_1)= \ker g_1, \mbox{ and}\] \[\ker\psi_1 \cap \ker f_1 \subseteq \im f_2\] then $(\psi_1)_*$ is an isomorphism.
\end{lem}
\begin{proof}
We have $(\psi_1)_*$ a surjection because $\psi_1$ is a surjection from $\ker f_1$ to $\ker g_1$. To see $(\psi_1)_*$ is injective, if $(\psi_1)_*(a_1+\im f_2)=0$ for $a_1\in \ker f_1$ then $\psi_1(a_1)\in \im g_2=\psi_1(\im f_2)$ so $\psi_1(a_1)=\psi_1(f_2(a_2))$ and thus $\psi_1(a_1-f_2(a_2))=0$ and also $f_1(a_1-f_2(a_2))=0$ so $a_1-f_2(a_2) \in \ker\psi_1\cap \ker f_1 \subseteq \im f_2$ so $a_1-f_2(a_2)=f_2(a_2')$ so $a_1=f_2(a_2+a_2')\in \im f_2$ as required.
\end{proof}

The following proposition gives us another way of looking at hypermap-homology which may be informative.

\begin{prop}\label{prop:H1asDarts}
\[H_1 = \frac{\ker\partial_1}{\im\partial_2} \cong \frac{\ker d_1}{\im d_2+\imath(\mathscr{E})}.\] 
\end{prop}
\begin{proof}
Consider the following diagram.
\[\begin{tikzcd}[ampersand replacement=\&]
\& \mathcal{F}\oplus\mathcal{E} \arrow{r}{d_2\oplus \imath} \arrow{d}{\pi_\mathcal{F}} \& \mathcal{W} \arrow{r}{d_1}\arrow{d}{p} \& \mathcal{V}\arrow{d}{id}\\
\&\mathcal{F} \arrow{r}{\partial_2} \& \mathcal{W}/\imath(\mathcal{E}) \arrow{r}{\partial_1} \& \mathcal{V}
\end{tikzcd}\]
This diagram is commutative because $p((d_2\oplus \imath)(f+e)=p(d_2(f))=\partial_2(f)=\partial_2(\pi_\mathcal{F}(f+e))$ and $\partial_1(p(w))=d_1(w)$. Next note that \[\im \partial_2=\im(p\circ d_2)=p(\im d_2)=p(\im d_2\oplus \imath)\] and using the fact that $p$ is surjective, \[\ker \partial_1=\{p(w) \colon \partial_1(p(w))=0\}=\{p(w) \colon d_1(w)=0\}=p(\ker d_1).\]

Finally, $\ker p=\imath(\mathcal{E}) \subseteq \im(d_2\oplus \imath)$ so by Lemma \ref{lem:inducedIsomorphism} we have $p_*$ an isomorphism.
\end{proof}

%
%

Next we consider how to relate hypermap-homology to the homology we can get by considering the embedded bipartite graph representation of the hypermap. The chain complex of this `classical homology'\index{homology!bipartite graph representation} is
\[\mathcal{F} \stackrel{\overline{d}_2}{\to} \mathcal{W} \stackrel{\overline{d}_1}{\to} \mathcal{V}\oplus\mathcal{E}\]
where $\overline{d}_2$ and $\overline{d}_1$ are defined by $\overline{d}_2(f) = \sum_{i\in f} w_i+w_{\alpha^{-1}(i)}$ and $\overline{d}_1(w_i)=v_{\ni i}+e_{\ni i}$ extended linearly.

\begin{prop}\label{prop:H1asClassical}
\[H_1 = \frac{\ker\partial_1}{\im\partial_2} \cong \frac{\ker \overline{d}_1}{\im \overline{d}_2}.\] 
\end{prop}
\begin{proof}
Consider the following diagram
\[\begin{tikzcd}[ampersand replacement=\&]
\& \mathcal{F} \arrow{r}{\partial_2} \arrow{d}{id} \& \mathcal{W}/\imath(\mathcal{E}) \arrow{r}{\partial_1}\arrow{d}{\mu} \& \mathcal{V}\arrow{d}{i_{\mathcal{V}}}\\
\&\mathcal{F} \arrow{r}{\overline{d}_2} \& \mathcal{W} \arrow{r}{\overline{d}_1} \& \mathcal{V}\oplus\mathcal{E}
\end{tikzcd}\]
where $\mu(w_i+\imath(\mathcal{E}))=w_i+w_{\alpha^{-1}(i)}$. This is well defined because \[\mu(\imath(e))=\sum_{i\in e}w_i+\sum_{i\in e}w_{\alpha^{-1}(i)}=0.\] The diagram is commutative because \[\mu(\partial_2(f))=\mu\left(\sum_{i\in f} w_i\right)=\sum_{i\in f} w_i+w_{\alpha^{-1}(i)}=\overline{d}_2(f)\] and 
\begin{align*}
\overline{d}_1(\mu(w_i+\imath(\mathcal{E})))&=\overline{d}_1(w_i+w_{\alpha^{-1}(i)})\\
&=v_{\ni i}+e_{\ni i}+v_{\ni\alpha^{-1}(i)}+e_{\ni \alpha^{-1}(i)}\\
&=v_{\ni i}+v_{\ni\alpha^{-1}(i)}\\
&=\partial_1(w_i+\imath(\mathcal{E}).
\end{align*}

We have $\im \overline{d}_2=\im(\mu\circ \partial_2)=\mu(\im\partial_2)$ and the earlier lemma gives us $\mu(\ker \partial_1)\subseteq \ker \overline{d}_1$. We now aim to prove the other inclusion.

If $x\in \ker\overline{d}_1$ then at each edge (i.e. square vertex of the bipartite graph representation) of the hypermap there are an even number of adjacent darts of $x$. Choose an ordering $w_{i_0},w_{i_1},\dots,w_{i_{2m-1}}$ of the darts at each edge so that there are no other darts in $x$ as we go counterclockwise from $w_{i_{2k}}$ to $w_{i_{2k+1}}$. Then if the darts in between $w_{i_{2k}}$ and $w_{i_{2k+1}}$ are $w_{l_1},\ldots,w_{l_a}$ we have $\mu(w_i+w_{l_1}+\dots+w_{l_a}+\imath(\mathcal{E}))=w_i+w_j$. Now by linearity, since all darts in $x$ are adjacent to one edge of the hypermap, this gives us a recipe to write $x=\mu(w+\imath(\mathcal{E}))$ for some $w\in \mathcal{W}$.

Finally, note that if $w=\sum_{i\in I} w_i$ for some index set $I$ then $\mu(w+\imath(\mathcal{E}))=\mu(\sum_{i\in I} w_i+\imath(\mathcal{E}))=\sum_{i\in I} (w_i + w_{\alpha^{-1}(i)})$. Thus if $\mu(w+\imath(\mathcal{E}))=0$ then $w$ is fixed by $\alpha^{-1}$ and must be in $\imath(\mathcal{E})$ so in fact $\mu$ is injective. Thus Lemma \ref{lem:inducedIsomorphism} shows that $\mu_*$ is an isomorphism.
\end{proof}

This gives us three ways to think about $H_1$:
\begin{enumerate}
\item $\partial_1$-cycles in $\mathcal{W}/\imath(\mathcal{E})$ modulo $\partial_2$-boundaries,
\item $d_1$-cycles in $\mathcal{W}$ modulo $d_2$-boundaries and $\imath$-boundaries, or
\item $\overline{d}_1$-cycles in $\mathcal{W}$ modulo $\overline{d}_2$-boundaries.
\end{enumerate}

\section{Hypermap-homology codes}\label{sec:hypHomCodes}

Choose a basis for $\mathcal{W}/\imath(\mathcal{E})$ (the bases for $\mathcal{F}$ and $\mathcal{V}$ have already been fixed).

Let $H_X=[\partial_1]$ and $H_Z^T=[\partial_2]$. Then we have $H_XH_Z^T=[\partial_1\circ\partial_2]=[0]=0$ so we can create a CSS code from $H_X$ and $H_Z$. Now $H_X$ is a $|V|\times(|W|-|E|)$ matrix and $H_Z$ is a $|F|\times(|W|-|E|)$ matrix. By the theory of CSS codes this code has parameters $[|W|-|E|,2g,D]$ with
\[D=\min\{\wt(c) \colon c \in (C_X \setminus C_Z^\perp) \cup (C_Z \setminus C_X^\perp)\}\]
where $C_X=\ker(H_X)$, $C_X^\perp=\im(H_X^T)$, $C_Z=\ker(H_Z)$ and $C_Z^\perp=\im(H_Z^T)$.

\begin{ex}
In this example we create a hypermap-homology code from example \ref{ex:squareOctagon}. Since $\mathcal{W}/\imath(\mathcal{E})=\langle w_1,w_2,w_3,\dots,w_{24} \mid w_1+w_2+w_3,\dots,w_{22}+w_{23}+w_{24}\rangle$ we may take as a basis $w_1,w_2,w_4,w_5,\dots,w_{23}$. Choosing this basis we can calculate
\[H_X=\begin{bmatrix}
1 & 1 & 0 & 0 & 0 & 0 & 0 & 0 & 0 & 0 & 0 & 0 & 0 & 1 & 1 & 0 \\
0 & 1 & 0 & 0 & 1 & 0 & 0 & 0 & 0 & 1 & 0 & 0 & 0 & 0 & 0 & 0 \\
1 & 0 & 0 & 0 & 0 & 0 & 0 & 1 & 1 & 1 & 0 & 0 & 0 & 0 & 0 & 0 \\
0 & 0 & 1 & 1 & 0 & 0 & 0 & 0 & 0 & 0 & 1 & 0 & 0 & 0 & 0 & 1 \\
0 & 0 & 0 & 1 & 0 & 0 & 0 & 0 & 0 & 0 & 0 & 1 & 1 & 0 & 0 & 0 \\
0 & 0 & 1 & 0 & 1 & 1 & 1 & 1 & 0 & 0 & 0 & 0 & 0 & 0 & 0 & 0 \\
0 & 0 & 0 & 0 & 0 & 1 & 1 & 0 & 0 & 0 & 1 & 1 & 0 & 0 & 0 & 0 \\
0 & 0 & 0 & 0 & 0 & 0 & 0 & 0 & 1 & 0 & 0 & 0 & 1 & 1 & 1 & 1
\end{bmatrix}\]
and
\[H_Z=\begin{bmatrix}
1 & 0 & 1 & 1 & 0 & 0 & 0 & 1 & 0 & 0 & 0 & 0 & 1 & 1 & 0 & 0 \\
0 & 1 & 1 & 0 & 1 & 0 & 0 & 0 & 0 & 0 & 0 & 0 & 0 & 0 & 1 & 1 \\
1 & 1 & 0 & 0 & 0 & 0 & 0 & 0 & 0 & 1 & 0 & 0 & 0 & 0 & 0 & 0 \\
0 & 0 & 0 & 1 & 0 & 0 & 0 & 0 & 0 & 0 & 1 & 1 & 0 & 0 & 0 & 0 \\
0 & 0 & 0 & 0 & 0 & 1 & 1 & 0 & 0 & 0 & 0 & 0 & 0 & 0 & 0 & 0 \\
0 & 0 & 0 & 0 & 1 & 1 & 0 & 0 & 1 & 1 & 1 & 0 & 0 & 0 & 0 & 1 \\
0 & 0 & 0 & 0 & 0 & 0 & 1 & 1 & 1 & 0 & 0 & 1 & 1 & 0 & 0 & 0 \\
0 & 0 & 0 & 0 & 0 & 0 & 0 & 0 & 0 & 0 & 0 & 0 & 0 & 1 & 1 & 0
\end{bmatrix}.\]

In both matrices the columns are labelled by $w_1,w_2,w_4,w_5,w_7,w_8,\dots,w_{23}$. In $H_X$ the rows are labelled by vertices $v_1,\dots,v_8$ while in $H_Z$ the rows are labelled by faces $f_1,\dots,f_8$. We give a few examples of how to compute these matrices:
\[\partial_1(w_{10})=v_{\ni 10} + v_{\ni 12}=v_6+v_7\]
\[\partial_2(f_1)=w_1+w_{11}+w_6+w_{21}=w_1+w_4+w_5+w_{11}+w_{19}+w_{20}.\]

By computer search (see Section \ref{sec:OOsoftware}) we calculate $D=2$ so this is a $[16,2,2]$ code.

\end{ex}

\begin{ex}
The hypermap in Figure \ref{fig:hexagon} has 48 darts, 16 vertices, 16 edges and 16 faces. We choose a basis for $\mathcal{W}/\imath(\mathcal{E})$ corresponding to all the darts except the ones to the bottom right of each edge. Computer search shows that the minimum distance is 3 and thus this is a $[32,2,3]$ code.

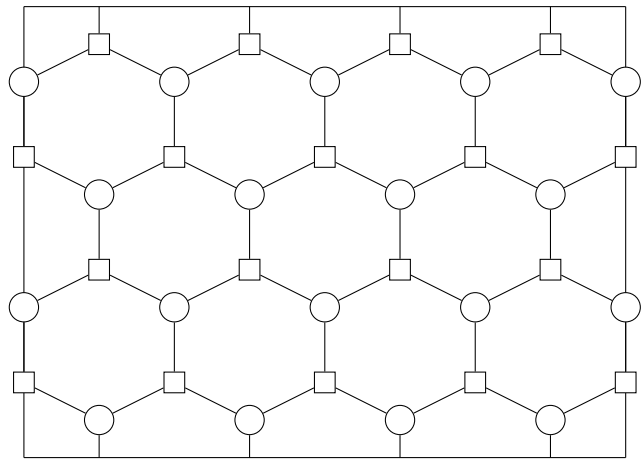
\begin{figure}[p!]
\centering
\begin{tikzpicture}
	\begin{pgfonlayer}{nodelayer}
		\node [style=edge] (0) at (-7, 6) {};
		\node [style=edge] (1) at (-5, 6) {};
		\node [style=edge] (2) at (-3, 6) {};
		\node [style=edge] (3) at (-1, 6) {};
		\node [style=vertex] (4) at (-8, 5.5) {};
		\node [style=vertex] (5) at (-6, 5.5) {};
		\node [style=vertex] (6) at (-4, 5.5) {};
		\node [style=vertex] (7) at (-2, 5.5) {};
		\node [style=vertex] (8) at (0, 5.5) {};
		\node [style=edge] (9) at (-8, 4.5) {};
		\node [style=edge] (10) at (-6, 4.5) {};
		\node [style=edge] (11) at (-4, 4.5) {};
		\node [style=edge] (12) at (-2, 4.5) {};
		\node [style=edge] (13) at (0, 4.5) {};
		\node [style=vertex] (14) at (-7, 4) {};
		\node [style=vertex] (15) at (-5, 4) {};
		\node [style=vertex] (16) at (-3, 4) {};
		\node [style=vertex] (17) at (-1, 4) {};
		\node [style=edge] (18) at (-7, 3) {};
		\node [style=edge] (19) at (-5, 3) {};
		\node [style=edge] (20) at (-3, 3) {};
		\node [style=edge] (21) at (-1, 3) {};
		\node [style=vertex] (22) at (-8, 2.5) {};
		\node [style=vertex] (23) at (-6, 2.5) {};
		\node [style=vertex] (24) at (-4, 2.5) {};
		\node [style=vertex] (25) at (-2, 2.5) {};
		\node [style=vertex] (26) at (0, 2.5) {};
		\node [style=edge] (27) at (-8, 1.5) {};
		\node [style=edge] (28) at (-6, 1.5) {};
		\node [style=edge] (29) at (-4, 1.5) {};
		\node [style=edge] (30) at (-2, 1.5) {};
		\node [style=edge] (31) at (0, 1.5) {};
		\node [style=vertex] (32) at (-7, 1) {};
		\node [style=vertex] (33) at (-5, 1) {};
		\node [style=vertex] (34) at (-3, 1) {};
		\node [style=vertex] (35) at (-1, 1) {};
		\node [style=none] (36) at (-8, 6.5) {};
		\node [style=none] (37) at (-7, 6.5) {};
		\node [style=none] (38) at (-5, 6.5) {};
		\node [style=none] (39) at (-3, 6.5) {};
		\node [style=none] (40) at (-1, 6.5) {};
		\node [style=none] (41) at (-1, 0.5) {};
		\node [style=none] (42) at (-3, 0.5) {};
		\node [style=none] (43) at (-5, 0.5) {};
		\node [style=none] (44) at (-7, 0.5) {};
		\node [style=none] (45) at (0, 6.5) {};
		\node [style=none] (46) at (0, 0.5) {};
		\node [style=none] (47) at (-8, 0.5) {};
	\end{pgfonlayer}
	\begin{pgfonlayer}{edgelayer}
		 \draw (4) to (0);
		\draw (0) to (5);
		\draw (5) to (10);
		\draw (10) to (14);
		\draw (14) to (9);
		\draw (4) to (9);
		\draw (5) to (1);
		\draw (0) to (37.center);
		\draw (1) to (38.center);
		\draw (1) to (6);
		\draw (6) to (11);
		\draw (11) to (15);
		\draw (15) to (10);
		\draw (6) to (2);
		\draw (2) to (7);
		\draw (7) to (12);
		\draw (12) to (16);
		\draw (16) to (11);
		\draw (7) to (3);
		\draw (3) to (8);
		\draw (8) to (13);
		\draw (13) to (17);
		\draw (17) to (12);
		\draw (14) to (18);
		\draw (18) to (22);
		\draw (18) to (23);
		\draw [in=90, out=-90] (23) to (28);
		\draw (28) to (32);
		\draw (32) to (27);
		\draw (27) to (22);
		\draw (23) to (19);
		\draw (19) to (15);
		\draw (20) to (24);
		\draw (24) to (19);
		\draw (20) to (16);
		\draw (17) to (21);
		\draw (21) to (25);
		\draw (25) to (20);
		\draw (21) to (26);
		\draw (3) to (40.center);
		\draw (2) to (39.center);
		\draw (31) to (26);
		\draw (31) to (35);
		\draw (35) to (41.center);
		\draw (35) to (30);
		\draw (30) to (25);
		\draw (24) to (29);
		\draw (29) to (33);
		\draw (33) to (28);
		\draw (32) to (44.center);
		\draw (33) to (43.center);
		\draw (34) to (42.center);
		\draw (34) to (29);
		\draw (34) to (30);
		\draw (36.center) to (45.center);
		\draw (45.center) to (46.center);
		\draw (47.center) to (46.center);
		\draw (47.center) to (36.center);
	\end{pgfonlayer}
\end{tikzpicture}
\caption[Hexagon hypermap]{A hypermap with hexagonal faces, embedded on a torus.}
\label{fig:hexagon}
\end{figure}
\end{ex}

\begin{ex}
The hypermap in Figure \ref{fig:4x4square} has 32 darts, 8 vertices, 8 edges and 16 faces. The dual hypermap can be seen in Figure \ref{fig:4x4squareDual}. We choose a basis for $\mathcal{W}/\imath(\mathcal{E})$ corresponding to all the darts except the ones below each edge. Computer search shows that the minimum distance is 4 and thus this is a $[24,2,4]$ code. Compare this to the toric code of minimum distance 4 which is a $[32,2,4]$ code.
\begin{figure}[p!]
\centering
\begin{tikzpicture}
	\begin{pgfonlayer}{nodelayer}
		\node [style=vertex] (0) at (-7, 6) {};
		\node [style=vertex] (1) at (-3, 6) {};
		\node [style=vertex] (2) at (-5, 4) {};
		\node [style=vertex] (3) at (-1, 4) {};
		\node [style=vertex] (4) at (-7, 2) {};
		\node [style=vertex] (5) at (-3, 2) {};
		\node [style=vertex] (6) at (-5, 0) {};
		\node [style=vertex] (7) at (-1, 0) {};
		\node [style=edge] (8) at (-5, 6) {};
		\node [style=edge] (9) at (-1, 6) {};
		\node [style=edge] (10) at (-7, 4) {};
		\node [style=edge] (11) at (-3, 4) {};
		\node [style=edge] (12) at (-5, 2) {};
		\node [style=edge] (13) at (-1, 2) {};
		\node [style=edge] (14) at (-7, 0) {};
		\node [style=edge] (15) at (-3, 0) {};
		\node [style=none] (16) at (-8, 6) {};
		\node [style=none] (17) at (-7, 7) {};
		\node [style=none] (18) at (-8, 7) {};
		\node [style=none] (19) at (-8, 4) {};
		\node [style=none] (20) at (-8, 2) {};
		\node [style=none] (21) at (-8, 0) {};
		\node [style=none] (22) at (-8, -1) {};
		\node [style=none] (23) at (-7, -1) {};
		\node [style=none] (24) at (-5, -1) {};
		\node [style=none] (25) at (-3, -1) {};
		\node [style=none] (26) at (-1, -1) {};
		\node [style=none] (27) at (0, -1) {};
		\node [style=none] (28) at (0, 0) {};
		\node [style=none] (29) at (0, 2) {};
		\node [style=none] (30) at (0, 4) {};
		\node [style=none] (31) at (0, 6) {};
		\node [style=none] (32) at (0, 7) {};
		\node [style=none] (33) at (-1, 7) {};
		\node [style=none] (34) at (-3, 7) {};
		\node [style=none] (35) at (-5, 7) {};
	\end{pgfonlayer}
	\begin{pgfonlayer}{edgelayer}
		\draw (18.center) to (22.center);
		\draw (22.center) to (27.center);
		\draw (27.center) to (32.center);
		\draw (32.center) to (18.center);
		\draw (17.center) to (0);
		\draw (0) to (16.center);
		\draw (0) to node[below]{1} (8);
		\draw (8) to node[left]{2} (35.center);
		\draw (34.center) to (1);
		\draw (1) to node[above]{3} (8);
		\draw (1) to node[below]{5} (9);
		\draw (9) to node[left]{6} (33.center);
		\draw (9) to node[above]{7} (31.center);
		\draw (9) to node[right]{8} (3);
		\draw (3) to (30.center);
		\draw (3) to node[above]{15} (11);
		\draw (11) to node[left]{14} (1);
		\draw (11) to node[below]{13} (2);
		\draw (2) to node[right]{4} (8);
		\draw [in=0, out=180] (2) to node[above]{11} (10);
		\draw (10) to node[left]{10} (0);
		\draw (10) to node[below]{9} (19.center);
		\draw (4) to (20.center);
		\draw (4) to node[right]{12} (10);
		\draw (4) to node[below]{17} (12);
		\draw [in=-90, out=90] (12) to node[left]{18} (2);
		\draw (11) to node[right]{16} (5);
		\draw (15) to node[left]{30} (5);
		\draw (15) to node[below]{29} (6);
		\draw (6) to node[right]{20} (12);
		\draw (12) to node[above]{19} (5);
		\draw (5) to node[below]{21} (13);
		\draw (13) to node[right]{24} (7);
		\draw (7) to node[above]{31} (15);
		\draw (13) to node[above]{23} (29.center);
		\draw [in=-90, out=90] (13) to node[left]{22} (3);
		\draw (4) to node[left]{26} (14);
		\draw (14) to node[below]{25} (21.center);
		\draw (14) to node[right]{28} (23.center);
		\draw (6) to (24.center);
		\draw (6) to node[above]{27} (14);
		\draw (15) to node[right]{32} (25.center);
		\draw (7) to (26.center);
		\draw (7) to (28.center);
	\end{pgfonlayer}
\end{tikzpicture}
\caption[$4\times4$ square hypermap]{A hypermap with square faces, embedded on a torus.}
\label{fig:4x4square}
\end{figure}
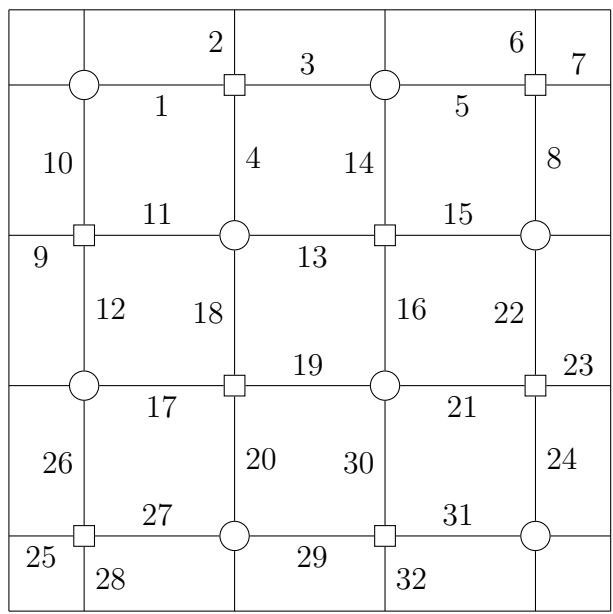

\begin{figure}[p!]
\centering
\begin{tikzpicture}[x=1.8 cm, y=1.8 cm]
	\begin{pgfonlayer}{nodelayer}
		\node [style=vertex] (0) at (-7, 6) {};
		\node [style=vertex] (1) at (-3, 6) {};
		\node [style=vertex] (2) at (-5, 4) {};
		\node [style=vertex] (3) at (-1, 4) {};
		\node [style=vertex] (4) at (-7, 2) {};
		\node [style=vertex] (5) at (-3, 2) {};
		\node [style=vertex] (6) at (-5, 0) {};
		\node [style=vertex] (7) at (-1, 0) {};
		\node [style=edge] (8) at (-5, 6) {};
		\node [style=edge] (9) at (-1, 6) {};
		\node [style=edge] (10) at (-7, 4) {};
		\node [style=edge] (11) at (-3, 4) {};
		\node [style=edge] (12) at (-5, 2) {};
		\node [style=edge] (13) at (-1, 2) {};
		\node [style=edge] (14) at (-7, 0) {};
		\node [style=edge] (15) at (-3, 0) {};
		\node [style=none] (16) at (-8, 6) {};
		\node [style=none] (17) at (-7, 7) {};
		\node [style=dualVertex] (18) at (-8, 7) {};
		\node [style=none] (19) at (-8, 4) {};
		\node [style=none] (20) at (-8, 2) {};
		\node [style=none] (21) at (-8, 0) {};
		\node [style=dualVertex] (22) at (-8, -1) {};
		\node [style=none] (23) at (-7, -1) {};
		\node [style=none] (24) at (-5, -1) {};
		\node [style=none] (25) at (-3, -1) {};
		\node [style=none] (26) at (-1, -1) {};
		\node [style=none] (27) at (0, -1) {};
		\node [style=none] (28) at (0, 0) {};
		\node [style=none] (29) at (0, 2) {};
		\node [style=none] (30) at (0, 4) {};
		\node [style=none] (31) at (0, 6) {};
		\node [style=dualVertex] (32) at (0, 7) {};
		\node [style=none] (33) at (-1, 7) {};
		\node [style=none] (34) at (-3, 7) {};
		\node [style=none] (35) at (-5, 7) {};
		\node [style=dualVertex] (36) at (-6, 5) {};
		\node [style=dualVertex] (37) at (-4, 5) {};
		\node [style=dualVertex] (38) at (-2, 5) {};
		\node [style=dualVertex] (39) at (-2, 3) {};
		\node [style=dualVertex] (40) at (-4, 3) {};
		\node [style=dualVertex] (41) at (-6, 3) {};
		\node [style=dualVertex] (42) at (-6, 1) {};
		\node [style=dualVertex] (43) at (-4, 1) {};
		\node [style=dualVertex] (44) at (-2, 1) {};
		\node [style=dualVertex] (45) at (-6, -1) {};
		\node [style=dualVertex] (46) at (-4, -1) {};
		\node [style=dualVertex] (47) at (-2, -1) {};
		\node [style=dualVertex] (48) at (0, -1) {};
		\node [style=dualVertex] (49) at (0, 1) {};
		\node [style=dualVertex] (50) at (0, 3) {};
		\node [style=dualVertex] (51) at (0, 5) {};
		\node [style=dualVertex] (52) at (-8, 5) {};
		\node [style=dualVertex] (53) at (-8, 3) {};
		\node [style=dualVertex] (54) at (-8, 1) {};
		\node [style=dualVertex] (55) at (-6, 7) {};
		\node [style=dualVertex] (56) at (-4, 7) {};
		\node [style=dualVertex] (57) at (-2, 7) {};
	\end{pgfonlayer}
	\begin{pgfonlayer}{edgelayer}
		\draw (18) to (22);
		\draw (22) to (27.center);
		\draw (27.center) to (32);
		\draw (32) to (18);
		\draw (17.center) to (0);
		\draw (0) to (16.center);
		\draw (0) to node[below]{1} (8);
		\draw (8) to node[left]{2} (35.center);
		\draw (34.center) to (1);
		\draw (1) to node[above]{3} (8);
		\draw (1) to node[below]{5} (9);
		\draw (9) to node[left]{6} (33.center);
		\draw (9) to node[above]{7} (31.center);
		\draw (9) to node[right]{8} (3);
		\draw (3) to (30.center);
		\draw (3) to node[above]{15} (11);
		\draw (11) to node[left]{14} (1);
		\draw (11) to node[below]{13} (2);
		\draw (2) to node[right]{4} (8);
		\draw [in=0, out=180] (2) to node[above]{11} (10);
		\draw (10) to node[left]{10} (0);
		\draw (10) to node[below]{9} (19.center);
		\draw (4) to (20.center);
		\draw (4) to node[right]{12} (10);
		\draw (4) to node[below]{17} (12);
		\draw (12) to node[left]{18} (2);
		\draw (11) to node[right]{16} (5);
		\draw (15) to node[left]{30} (5);
		\draw (15) to node[below]{29} (6);
		\draw (6) to node[right]{20} (12);
		\draw (12) to node[above]{19} (5);
		\draw (5) to node[below]{21} (13);
		\draw (13) to node[right]{24} (7);
		\draw (7) to node[above]{31} (15);
		\draw (13) to node[above]{23} (29.center);
		\draw (13) to node[left]{22} (3);
		\draw (4) to node[left]{26} (14);
		\draw (14) to node[below]{25} (21.center);
		\draw (14) to node[right]{28} (23.center);
		\draw (6) to (24.center);
		\draw (6) to node[above]{27} (14);
		\draw (15) to node[right]{32} (25.center);
		\draw (7) to (26.center);
		\draw (7) to (28.center);
		\draw [color=red] (36) to node[left]{1} (8);
		\draw [color=red, in=45, out=-135] (36) to node[right]{11} (10);
		\draw [color=red] (37) to node[left]{4} (8);
		\draw [color=red] (37) to node[right]{14} (11);
		\draw [color=red] (11) to node[right]{15} (38);
		\draw [color=red] (38) to node[left]{5} (9);
		\draw [color=red] (9) to node[left]{8} (51);
		\draw [color=red] (11) to node[left]{13} (40);
		\draw [color=red] (40) to node[right]{19} (12);
		\draw [color=red] (12) to node[right]{18} (41);
		\draw [color=red] (41) to node[left]{12} (10);
		\draw [color=red] (11) to node[left]{16} (39);
		\draw [color=red] (39) to node[right]{22} (13);
		\draw [color=red] (13) to node[right]{23} (50);
		\draw [color=red] (8) to node[below]{3} (56);
		\draw [color=red] (9) to node[above]{6} (57);
		\draw [color=red] (9) to node[right]{7} (32);
		\draw [color=red] (10) to node[right]{10} (52);
		\draw [color=red] (10) to node[left]{9} (53);
		\draw [color=red] (12) to node[left]{17} (42);
		\draw [color=red] (12) to node[left]{20} (43);
		\draw [color=red] (13) to node[left]{21} (44);
		\draw [color=red] (13) to node[left]{24} (49);
		\draw [color=red] (15) to node[right]{30} (43);
		\draw [color=red] (15) to node[above]{29} (46);
		\draw [color=red] (15) to node[right]{31} (44);
		\draw [color=red] (45) to node[below]{28} (14);
		\draw [color=red] (14) to node[right]{27} (42);
		\draw [color=red] (14) to node[right]{26} (54);
		\draw [color=red] (14) to node[left]{25} (22);
		\draw [color=red] (8) to node[above]{2} (55);
		\draw [color=red, in=135, out=-45] (15) to node[below]{32} (47);
	\end{pgfonlayer}
\end{tikzpicture}
\caption[$4\times4$ square hypermap with dual]{A hypermap with square faces, embedded on a torus. The dual hypermap is shown in red.}
\label{fig:4x4squareDual}
\end{figure}
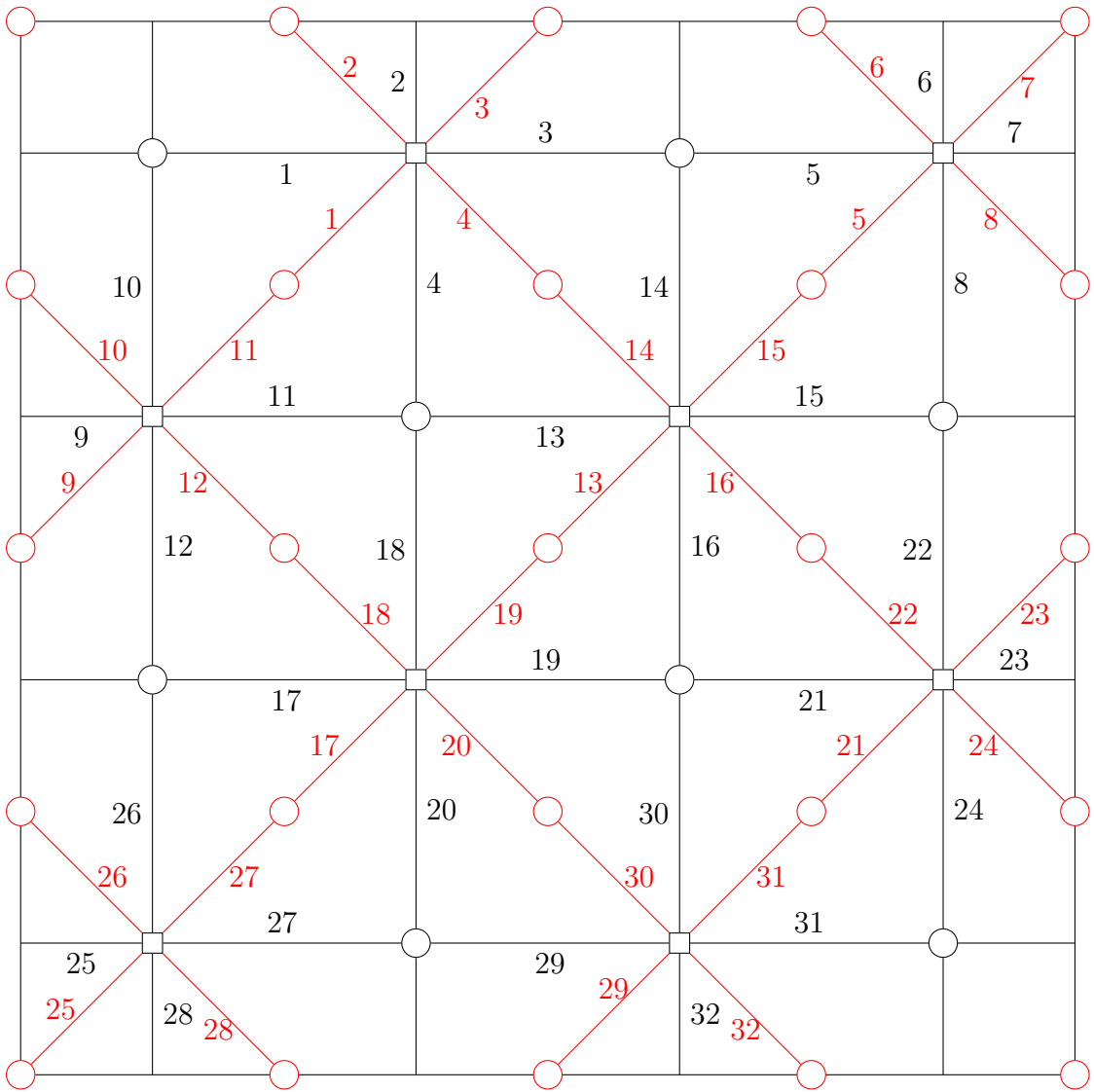

\end{ex}

\begin{ex}
The hypermap in Figure \ref{fig:DLCube} is based on a graph embedding from \cite{gagarin2003embeddings} which is described as the dual of the line graph of a cube. We choose a basis for $\mathcal{W}/\imath(\mathcal{E})$ corresponding to all the darts whose labels are not divisible by 6. Computer calculations show that this gives rise to a $[20,2,3]$ code. 
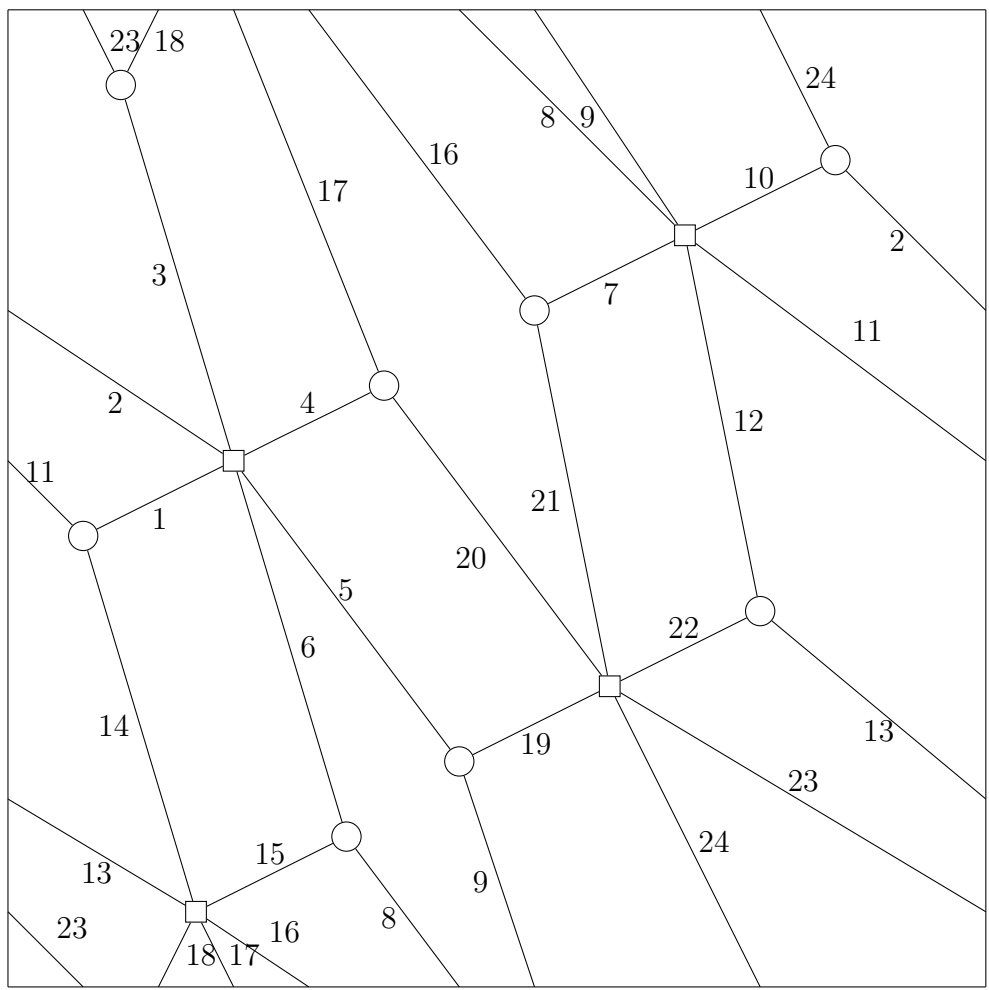
\begin{figure}[p!]
\centering
\begin{tikzpicture}[x=2 cm, y=2 cm]
	\begin{pgfonlayer}{nodelayer}
		\node [style=vertex] (0) at (-4, 0) {};
		\node [style=vertex] (1) at (-2, 1) {};
		\node [style=vertex] (2) at (-1, 1.5) {};
		\node [style=vertex] (3) at (1, 2.5) {};
		\node [style=edge] (4) at (-3, 0.5) {};
		\node [style=edge] (5) at (0, 2) {};
		\node [style=vertex] (6) at (-3.75, 3) {};
		\node [style=vertex] (7) at (0.5, -0.5) {};
		\node [style=vertex] (8) at (-1.5, -1.5) {};
		\node [style=vertex] (9) at (-2.25, -2) {};
		\node [style=edge] (10) at (-0.5, -1) {};
		\node [style=edge] (11) at (-3.25, -2.5) {};
		\node [style=none] (12) at (-4.5, 3.5) {};
		\node [style=none] (13) at (-4.5, -3) {};
		\node [style=none] (14) at (2, 3.5) {};
		\node [style=none] (15) at (2, -3) {};
		\node [style=none] (16) at (-4, 3.5) {};
		\node [style=none] (17) at (-3.5, 3.5) {};
		\node [style=none] (18) at (-3.5, -3) {};
		\node [style=none] (19) at (-3, -3) {};
		\node [style=none] (20) at (-4.5, 0.5) {};
		\node [style=none] (21) at (2, 1.5) {};
		\node [style=none] (22) at (-4.5, -1.75) {};
		\node [style=none] (23) at (2, -1.75) {};
		\node [style=none] (24) at (-1.5, 3.5) {};
		\node [style=none] (25) at (-1, 3.5) {};
		\node [style=none] (26) at (-1, -3) {};
		\node [style=none] (27) at (-1.5, -3) {};
		\node [style=none] (28) at (-4.5, -2.5) {};
		\node [style=none] (29) at (-4, -3) {};
		\node [style=none] (30) at (-3, 3.5) {};
		\node [style=none] (31) at (-2.5, 3.5) {};
		\node [style=none] (32) at (-2.5, -3) {};
		\node [style=none] (33) at (0.5, 3.5) {};
		\node [style=none] (34) at (0.5, -3) {};
		\node [style=none] (35) at (2, -2.5) {};
		\node [style=none] (36) at (-4.5, 1.5) {};
		\node [style=none] (37) at (2, 0.5) {};
	\end{pgfonlayer}
	\begin{pgfonlayer}{edgelayer}
		\draw (0) to node[below]{1} (4);
		\draw [in=-153, out=27] (4) to node[above]{4} (1);
		\draw (2) to node[below]{7} (5);
		\draw (5) to node[above]{10} (3);
		\draw (1) to node[below left]{20} (10);
		\draw (10) to node[left]{21} (2);
		\draw (10) to node[below]{19} (8);
		\draw (10) to node[above]{22} (7);
		\draw (9) to node[above]{15} (11);
		\draw (4) to node[left]{3} (6);
		\draw (4) to node[right]{6} (9);
		\draw (4) to node[above]{5} (8);
		\draw (12.center) to (14.center);
		\draw (14.center) to (15.center);
		\draw (15.center) to (13.center);
		\draw (13.center) to (12.center);
		\draw (7) to node[right]{12} (5);
		\draw (6) to node[right]{23} (16.center);
		\draw (6) to node[right]{18} (17.center);
		\draw (11) to node[right]{18} (18.center);
		\draw (11) to node[right]{17} (19.center);
		\draw (0) to node[above]{11} (20.center);
		\draw (3) to node[left]{2} (21.center);
		\draw (11) to node[below]{13} (22.center);
		\draw (23.center) to node[below]{13} (7);
		\draw (11) to node[left]{14} (0);
		\draw (5) to node[left]{9} (25.center);
		\draw (5) to node[left]{8} (24.center);
		\draw (9) to node[left]{8} (27.center);
		\draw (8) to node[left]{9} (26.center);
		\draw (29.center) to node[above right]{23} (28.center);
		\draw (1) to node[right]{17} (30.center);
		\draw (2) to node[right]{16} (31.center);
		\draw (11) to node[above right]{16} (32.center);
		\draw (3) to node[right]{24} (33.center);
		\draw (34.center) to node[right]{24} (10);
		\draw (10) to node[above]{23} (35.center);
		\draw (37.center) to node[above right]{11} (5);
		\draw (4) to node[below]{2} (36.center);
	\end{pgfonlayer}
\end{tikzpicture}
\caption[Dual(L(Cube)) hypermap]{A hypermap based on the dual of the line graph of a cube.}
\label{fig:DLCube}
\end{figure}

\end{ex}

\clearpage

\section{Finding weights given a special basis}\label{sec:specialBasis}
We choose, for each edge of the hypermap, one dart adjacent to that edge. Call this set of darts $S$, for `special darts'\index{special basis}. Then we choose a basis of $\mathcal{W}/\imath(\mathcal{E})$ which is $W\setminus S$. With a basis chosen we have a weight function on $\mathcal{W}/\imath(\mathcal{E})$ with the weight of $x+\imath(\mathcal{E})$ given by the number of nonzero basis vectors in its representation in this `special basis'.

If we look at the proof of Proposition \ref{prop:H1asClassical} then we see that $\mu \colon \ker\partial_1\setminus\im\partial_2 \to \ker\overline{d}_1\setminus\im\overline{d}_2$ is a bijection. Thus we can say that \[\operatorname{minwt}(C_X\setminus C_Z^\perp)=\operatorname{minwt}\{\mu^{-1}(x) \colon x\in \ker\overline{d}_1\setminus\im\overline{d}_2\}.\]
Now we seek to understand how we can find the weight of $\mu^{-1}(x)$ given $x$. We can choose an ordering of the darts around each edge using two cases. If the special dart adjacent to the edge is in $x$ then choose the pairing of darts so that the edge is at the end of a pair not the start. If the special dart is not in $x$ then we can choose an ordering of the darts so that the special dart is not between two paired darts. In both cases $\mu^{-1}(x)$ does not contain any special darts. Now $\mu^{-1}(x)$ is uniquely determined and we have chosen a representation of $\mu^{-1}(x)$ without any special darts so the weight of $\mu^{-1}(x)$ is just the number of darts in its sum. We see also that
\begin{align*}
&\operatorname{wt}(\mu^{-1}(x))=\frac{1}{2}\operatorname{wt}(x)+\\
&\text{\# of non-} x \text{ darts as we go counterclockwise around edges avoiding special darts.}
\end{align*}
One way to understand this is to imagine that at each of the darts that we skip as we go around an edge our cycle traverses that dart twice. Then every cycle in $\ker\overline{d}_1$ can be thought of as a cycle that does not skip darts around edges and $\mu^{-1}$ of that cycle is of weight half the weight of the cycle.

With the basis $W\setminus S$ chosen then we can identify vector spaces with their duals in the standard way and then the coboundary operators are given by transposes of boundary operators. If we look at the matrices $\delta_1=H_X^T$ and $\delta_2=H_Z$ then we see that $\delta_1$ takes vertices to non-special darts and $\delta_2$ takes non-special darts to faces so we will identify $(\mathcal{W}/\imath(\mathcal{E}))^*$ with $\mathcal{W}/\langle S\rangle$. Thus our hypermap-cohomology will be
\[\mathcal{V} \stackrel{\delta_1}{\to} \mathcal{W}/\langle S\rangle \stackrel{\delta_2}{\to} \mathcal{F}\]
where
\begin{align*}
\delta_1(v)&=\text{sum of non-special darts whose $\partial_1$-boundary contains $v$}\\
&=\text{map-boundary of face $v$ in dual hypermap without special darts}
\end{align*}
and
\begin{align*}
\delta_2(w_i)&=\text{faces whose $\partial_2$-boundaries contain $w_i$}\\
&=f_{\ni i}+f_{\ni \text{special dart of } e_{\ni i}}.
\end{align*}

We can also consider the classical cohomology of the bipartite graph representation of the hypermap. This is
\[\mathcal{V} \stackrel{\overline{\delta}_1}{\to} \mathcal{W} \stackrel{\overline{\delta}_2}{\to} \mathcal{F}\oplus\mathcal{E}\]
where \[\overline{\delta}_1(v)=\text{map-boundary of face $v$ in dual-hypermap}\]
and
\[\overline{\delta}_2(w_i)=f_{\ni i}+e_{\ni i}.\]

\begin{prop}\label{prop:cohomologyClassicalHomOfDual}
\[H^1 = \frac{\ker\delta_2}{\im\delta_1} \cong \frac{\ker \overline{\delta}_2}{\im \overline{\delta}_1}.\] 
\end{prop}
\begin{proof}
Consider the diagram
\[\begin{tikzcd}[ampersand replacement=\&]
\& \mathcal{V} \arrow{r}{\overline{\delta}_1} \arrow{d}{id} \& \mathcal{W} \arrow{r}{\overline{\delta}_2}\arrow{d}{\pi} \& \mathcal{E}\oplus\mathcal{F}\arrow{d}{\lambda}\\
\&\mathcal{V} \arrow{r}{\delta_1} \& \mathcal{W}/\langle S\rangle \arrow{r}{\delta_2} \& \mathcal{F}
\end{tikzcd}\]
where $\pi$ is the projection $\mathcal{W}\to\mathcal{W}/\langle S\rangle$ and $\lambda$ is defined by $\lambda(f_i)=f_i$ and $\lambda(e_i)=f_{\ni\text{special dart of } e_i}$.

This diagram is commutative because
\[\delta_1(v)=\text{map-boundary of face $v$ in dual hypermap without special darts}=\pi(\overline{\delta}_1(v))\] and
\begin{align*}
\lambda(\overline{\delta}_2(w_i))&=\lambda(f_{\ni i}+e_{\ni i})\\
&=f_{\ni i}+f_{\ni\text{special dart of } e_{\ni i}}\\
&=\delta_2(\pi(w_i))
\end{align*}
because $f_{\ni i}+f_{\ni\text{special dart of } e_{\ni i}}=0$ precisely when $w_i \in S$.

Now $\im \delta_1=\im(\pi\circ\overline{\delta}_1)=\pi(\im\overline{\delta}_1)$ and as usual we know that $\pi(\ker\delta_2)\subseteq\ker\delta_2$. We now show the opposite inclusion. If $x\in \ker\delta_2$, define $w\in \mathcal{W}$ to contain all the darts of $x$ and then for each edge of the hypermap, if there are an odd number of darts in $x$ adjacent to that edge, add the special dart adjacent to that edge to $w$. Clearly $\pi(w)=x$, we now claim that $w\in \ker\overline{\delta_2}$. To see this, note that
\[\lambda(\overline{\delta}_2(w))=\delta_2(\pi(w))=\delta_2(x)=0\]
and
\[\overline{\delta}_2(w)=\overline{\delta}_2\left(\sum_{i\in I} w_i\right)=\sum_{i\in I}\left(f_{\ni i}+e_{\ni i}\right)=\sum_{i\in I}f_{\ni i}\]
because each edge has an even number of darts of $w$ adjacent. Thus $\lambda(\overline{\delta}_2(w))=\overline{\delta}_2(w)$ and $w\in \ker\overline{\delta_2}$ as desired.

Finally if $w\in \ker\pi \cap\ker\overline{\delta}_2$ then $w$ is a sum of special darts which is a classical cycle in the dual hypermap. In particular each edge is adjacent to an even number of darts in $w$. But each edge is adjacent to only one special dart in the whole hypermap so we must have $w=0$. Thus $\ker\pi \cap\ker\overline{\delta}_2 \subseteq \im\overline{\delta}_1$ and Lemma \ref{lem:inducedIsomorphism} tells us that $\pi_*$ is an isomorphism.
\end{proof}

\begin{cor}
There is a bijection
\[\pi \colon \ker \overline{\delta}_2\setminus\im \overline{\delta}_1 \to \ker\delta_2\setminus\im\delta_1\]
which takes non-boundary cycles of hypermap-cohomology to non-boundary classical cycles of the dual hypermap.
\end{cor}

\begin{cor}
$\operatorname{minwt}(C_Z\setminus C_X^\perp)$ is given by the minimum weight of a non-boundary classical cycle in the dual hypermap where non-special darts have weight 1 and special darts have weight 0.
\end{cor}

\begin{ex}\label{ex:mbymsquaregrid}
Take an $m\times m$ square grid hypermap embedded on the torus with $m$ even and the special darts chosen to be the darts underneath each edge  (see Figure \ref{fig:4x4squareDual} for the hypermap and its dual in the $m=4$ case, where the special darts are those with labels divisible by 4). Then we have an $[N,k,d]$ code with $N=n-|E|=2m^2-(1/2)m^2=(3/2)m^2$ and $k=2g=2$.

For $\minwt(C_Z\setminus C_X^\perp)$ notice that there is a path in the dual hypermap of length $2m$ with half of the darts special darts. Also there are no paths in the dual hypermap of classical weight less than $2m$ so $\minwt(C_Z\setminus C_X^\perp)=m$.

For $\minwt(C_X\setminus C_Z^\perp)$ notice that a horizontal or vertical classical cycle $x\in\ker\overline{d}_1\setminus\im\overline{d}_2$ in the hypermap has $\wt(\mu^{-1}(x))=m$. To see this is the minimum weight, if $x\in\ker\overline{d}_1\setminus\im\overline{d}_2$ then $x$ is not a boundary so without loss of generality $x$ has at least $m$ horizontal darts. But $\mu(\mu^{-1}(x))=x$ and $\mu$ takes $w_i$ to $w_i$ plus $w_i$ rotated around an edge. Therefore $\mu$ acting on one dart can only ever lead to 1 horizontal dart so we must in fact have $\wt(\mu^{-1}(x))\geq m$.

Thus we have a $[(3/2)m^2,2,m]$ code with $k=2$ and $d=\sqrt{\frac{2}{3}N}$, so $kd^2=\frac{4}{3}N$.
\end{ex}

\section{Future work}
Many questions about hypermap-homology codes remain unanswered. We mention a few of them here:
\begin{enumerate}
\item Must hypermap-homology codes with a special basis satisfy $kd^2<cN$ for some constant $c$?
\item Can we analyze hypermap-homology codes with a non-special basis?
\item Can we find families of hypermaps which lead to better performance than the square grid hypermap?
\item Can we construct `planar hypermap-homology codes'?
\item Can we analyze the performance of randomly generated hypermap-homology codes?
\end{enumerate}

\appendix
\chapter{Software}
\section{Matlab software for hypermap-homology codes}\label{sec:OOsoftware}
We implemented an object-oriented package in Matlab to work with hypermap-homology codes. This includes classes for permutations, hypermaps and CSS codes. The use of this package in this dissertation is restricted to determining the parameters $[n,k,d]$ of a code generated by an input hypermap $(\sigma,\alpha)$. We now describe the simple algorithms we use to do this.

First we check if $\langle \sigma,\alpha\rangle\leq S_n$ is transitive. This algorithm comes from \cite{holt2005handbook}.

\algsetup{indent=1em}
\begin{algorithm}[h]
	\caption{CheckTransitive($\sigma$,$\alpha$)}
	\begin{algorithmic}[1]
		\REQUIRE $\sigma ,\alpha \in S_n$
		\ENSURE Output whether $\langle \sigma,\alpha\rangle$ is transitive
		\medskip
		\STATE Set $\mathrm{Orbit}=\{1\}$
		\REPEAT
			\FOR {$i=1,\dots,n$}
				\STATE Set $\mathrm{Orbit}=\mathrm{Orbit}\cup\{\sigma(i), \alpha(i)\}$
			\ENDFOR
		\UNTIL {$\mathrm{Orbit}$ did not change in last iteration}
		\IF {$|\mathrm{Orbit}|=n$}
			\RETURN `Yes'
		\ELSE
			\RETURN `No'
		\ENDIF
	\end{algorithmic}
\end{algorithm}

After finding the matrices $H_X$ and $H_Z$ using the definitions in Chapter \ref{ch:hypermapHomologyCodes} we can then determine the parameters of the associated code.

\algsetup{indent=1em}
\begin{algorithm}[ht]
	\caption{FindParameters($H_X$,$H_Z$)}
	\begin{algorithmic}[1]
		\REQUIRE Binary matrices $H_X$ and $H_Z$ with $H_XH_Z^T=0$
		\ENSURE Output parameters $[n,k,d]$
		\medskip
		\STATE Set $n=\mbox{width of matrix } H_X$
		\STATE Find $\rank(H_X)$ and $\rank(H_Z)$ by putting these matrices in RREF
		\STATE Set $k=n-\rank(H_X)-\rank(H_Z)$
		\STATE Find generator matrices $G_X$, $G_Z$ by finding nullspace of $H_X$, $H_Z$
		\STATE Set $k_X$ to the height of $G_X$ and $k_Z$ to the height of $G_Z$
		\STATE Set $d=\infty$
		\FORALL {$u \in \F_2^{k_X}$}
			\STATE Set $c=G_X u$
			\IF {$0<\wt(c)<d$ and $G_Z^Tc\neq 0$}
				\STATE Set $d=\wt(c)$
			\ENDIF
		\ENDFOR
		\FORALL {$u \in \F_2^{k_Z}$}
			\STATE Set $c=G_Z u$
			\IF {$0<\wt(c)<d$ and $G_X^Tc\neq 0$}
				\STATE Set $d=\wt(c)$
			\ENDIF
		\ENDFOR
		\RETURN $n,k,d$
	\end{algorithmic}
\end{algorithm}

\printindex

\bibliography{mybib.bib}{}

\begin{thebibliography}{DKLP02}

\bibitem[BMD07]{bombin2007homological}
H.~Bombin and MA~Martin-Delgado.
\newblock Homological error correction: classical and quantum codes.
\newblock {\em Journal of mathematical physics}, 48(5):052105--052105, 2007.

\bibitem[BPT10]{bravyi2010tradeoffs}
S.~Bravyi, D.~Poulin, and B.~Terhal.
\newblock Tradeoffs for reliable quantum information storage in 2d systems.
\newblock {\em Physical review letters}, 104(5):50503, 2010.

\bibitem[CDZ11]{couvreur2011construction}
A.~Couvreur, N.~Delfosse, and G.~Z{\'e}mor.
\newblock A construction of quantum ldpc codes from cayley graphs.
\newblock In {\em Information Theory Proceedings (ISIT), 2011 IEEE
  International Symposium on}, pages 643--647. IEEE, 2011.

\bibitem[CM92]{cori1992maps}
R.~Cori and A.~Machi.
\newblock Maps, hypermaps and their automorphisms: a survey i, ii, iii.
\newblock {\em Exposition. Math}, 10(5):403--427, 1992.

\bibitem[CP80]{cori1980complexity}
R.~Cori and J.G. Penaud.
\newblock The complexity of a planar hypermap and that of its dual.
\newblock {\em Annals of Discrete Mathematics}, 9:53--62, 1980.

\bibitem[CS96]{calderbank1996good}
A.R. Calderbank and P.W. Shor.
\newblock Good quantum error-correcting codes exist.
\newblock {\em Physical Review A}, 54(2):1098, 1996.

\bibitem[DCP10]{duclos2010fast}
G.~Duclos-Cianci and D.~Poulin.
\newblock Fast decoders for topological quantum codes.
\newblock {\em Physical review letters}, 104(5):50504, 2010.

\bibitem[DKLP02]{dennis2002topological}
E.~Dennis, A.~Kitaev, A.~Landahl, and J.~Preskill.
\newblock Topological quantum memory.
\newblock {\em Journal of Mathematical Physics}, 43:4452, 2002.

\bibitem[Fet11]{fetaya2011homological}
E.~Fetaya.
\newblock Homological error correcting codes and systolic geometry.
\newblock {\em Arxiv preprint arXiv:1108.2886}, 2011.

\bibitem[FML02]{freedman2002z2}
M.H. Freedman, D.A. Meyer, and F.~Luo.
\newblock Z2-systolic freedom and quantum codes.
\newblock {\em Mathematics of quantum computation, Chapman \& Hall/CRC}, pages
  287--320, 2002.

\bibitem[Gal62]{gallager1962low}
Robert Gallager.
\newblock Low-density parity-check codes.
\newblock {\em Information Theory, IRE Transactions on}, 8(1):21--28, 1962.

\bibitem[GKN03]{gagarin2003embeddings}
A.~Gagarin, W.~Kocay, and D.~Neilson.
\newblock Embeddings of small graphs on the torus.
\newblock {\em Cubo. v5}, pages 351--371, 2003.

\bibitem[Got97]{gottesman1997stabilizer}
D.~Gottesman.
\newblock Stabilizer codes and quantum error correction.
\newblock {\em Arxiv preprint quant-ph/9705052}, 1997.

\bibitem[Gur10]{guruswami2010lecture}
V.~Guruswami.
\newblock Course notes for introduction to coding theory.
\newblock {\em Carnegie Mellon University}, 2010.

\bibitem[Hat02]{hatcheralgebraic}
Allen Hatcher.
\newblock {\em Algebraic topology}.
\newblock Cambridge University Press, 2002.

\bibitem[HEO05]{holt2005handbook}
Derek~F Holt, Bettina Eick, and Eamonn~A O'Brien.
\newblock {\em Handbook of computational group theory}, volume~24.
\newblock Chapman and Hall/CRC, 2005.

\bibitem[Kit03]{kitaev2003fault}
A.Y. Kitaev.
\newblock Fault-tolerant quantum computation by anyons.
\newblock {\em Annals of Physics}, 303(1):2--30, 2003.

\bibitem[KP12]{kovalev2012improved}
A.A. Kovalev and L.P. Pryadko.
\newblock Improved quantum hypergraph-product ldpc codes.
\newblock {\em Arxiv preprint arXiv:1202.0928}, 2012.

\bibitem[LZZ04]{lando2004graphs}
S.K. Lando, A.K. Zvonkin, and D.B. Zagier.
\newblock {\em Graphs on surfaces and their applications}, volume 141.
\newblock Springer Verlag, 2004.

\bibitem[Mac03]{mackay2003information}
D.J.C. MacKay.
\newblock {\em Information theory, inference, and learning algorithms}.
\newblock Cambridge Univ Pr, 2003.

\bibitem[Maz11]{mazoit2011tree}
F.~Mazoit.
\newblock Tree-width of hypergraphs and surface duality.
\newblock {\em Journal of Combinatorial Theory, Series B}, 2011.

\bibitem[MMM04]{mackay2004sparse}
D.J.C. MacKay, G.~Mitchison, and P.L. McFadden.
\newblock Sparse-graph codes for quantum error correction.
\newblock {\em Information Theory, IEEE Transactions on}, 50(10):2315--2330,
  2004.

\bibitem[NC10]{nielsenquantum}
Michael~A Nielsen and Isaac~L Chuang.
\newblock {\em Quantum computation and quantum information}.
\newblock Cambridge university press, 2010.

\bibitem[RU08]{richardson2008modern}
T.J. Richardson and R.L. Urbanke.
\newblock {\em Modern coding theory}.
\newblock Cambridge Univ Pr, 2008.

\bibitem[SB12]{sarvepalli2012topological}
P.~Sarvepalli and K.R. Brown.
\newblock Topological subsystem codes from graphs and hypergraphs.
\newblock {\em arXiv preprint arXiv:1207.0479}, 2012.

\bibitem[TZ09]{tillich2009quantum}
J.P. Tillich and G.~Z{\'e}mor.
\newblock Quantum ldpc codes with positive rate and minimum distance
  proportional to $n^{1/2}$.
\newblock In {\em Information Theory, 2009. ISIT 2009. IEEE International
  Symposium on}, pages 799--803. IEEE, 2009.

\bibitem[Wal75]{walsh1975hypermaps}
TRS Walsh.
\newblock Hypermaps versus bipartite maps.
\newblock {\em Journal of Combinatorial Theory, Series B}, 18(2):155--163,
  1975.

\bibitem[Z{\'e}m09]{zemor2009cayley}
G.~Z{\'e}mor.
\newblock On cayley graphs, surface codes, and the limits of homological coding
  for quantum error correction.
\newblock {\em Coding and Cryptology}, pages 259--273, 2009.

\end{thebibliography}
\bibliographystyle{alpha}

\end{document}